\documentclass[a4]{article}
\usepackage{fullpage}
\usepackage{enumerate}
\usepackage{amsmath}
\usepackage{amssymb}
\usepackage{amsthm,amstext}
\usepackage{graphicx}
\usepackage{color}
\usepackage{subfigure}
\usepackage{psfrag}
\usepackage{auto-pst-pdf}
\usepackage{multirow}
\usepackage{mathtools}
\usepackage{units}

\title{Asymptotic-preserving Particle-In-Cell methods for the Vlasov-Maxwell system near quasi-neutrality}
\author{P. Degond$^\dagger$, F. Deluzet$^{\ddagger}$, D. Doyen$^\star$ \\[1em]
$^\dagger$Department of Mathematics, Imperial College London, \\
London SW7 2AZ, United Kingdom, \\
pdegond@imperial.ac.uk\\[1em]
$^\ddagger$Universit\'e de Toulouse; UPS, INSA, UT1, UTM, Institut de Math\'ematiques de Toulouse,\\
CNRS, Institut de Math´ematiques de Toulouse UMR 5219,
F-31062 Toulouse, France,\\
fabrice.deluzet@math.univ-toulouse.fr \\[1em]
$^\star$Universit\'e de Marne-la-Vall\'ee, Laboratoire d'Analyse et de Math\'ematiques Appliqu\'ees,\\
CNRS, Laboratoire d'Analyse et de Math\'ematiques Appliqu\'ees UMR 8050,\\
5, boulevard Descartes, Cit\'e Descartes - Champs-sur-Marne,
F-77454 Marne-la-Vallée, France,\\
david.doyen@univ-mlv.fr
}

\date{}

\def\Rset{\mathbb{R}}
\def\demi{\frac{1}{2}}

\newcommand{\E}{\mathcal{E}}

\newtheorem{remark}{Remark}[section]
\newtheorem{prop}{Proposition}[section]

\begin{document}


\maketitle

\begin{abstract}
In this article, we design Asymptotic-Preserving Particle-In-Cell methods for the Vlasov-Maxwell system in the quasi-neutral limit, this limit being characterized by a Debye length negligible compared to the space scale of the problem. These methods are consistent discretizations of the Vlasov-Maxwell system which, in the quasi-neutral limit, remain stable and are consistent with a quasi-neutral model (in this quasi-neutral model, the electric field is computed by means of a generalized Ohm law). The derivation of Asymptotic-Preserving methods is not straightforward since the quasi-neutral model is a singular limit of the Vlasov-Maxwell model. The key step is a reformulation of the Vlasov-Maxwell system which unifies the two models in a single set of equations with a smooth transition from one to another. As demonstrated in various and demanding numerical simulations, the Asymptotic-Preserving methods are able to treat efficiently both quasi-neutral plasmas and non-neutral plasmas, making them particularly well suited for complex problems involving dense plasmas with localized non-neutral regions.
\end{abstract}

\section{Introduction}

In a plasma, the Coulomb interaction between charged particles tends to restore the charge neutrality, while the thermal motion tends to disturb it. These opposing phenomena introduce a typical length of the separation between the electron ($n_e$) and the ion ($n_i$) densities, called the Debye length, and a typical oscillation period of the electrons, called the (electron) plasma period. These parameters depend essentially on the density and the thermal velocity of the particles. When the scales of interest are large compared to the Debye length, the charge separations may be neglected. In other words, the plasma may be assumed quasi-neutral. In that case, the Poisson equation is meaningless for the computation of the electric field\footnote{For the quasi-neutral regime investigated, this also applies to the Maxwell-Amp\`ere equation.} $E$, because this field varies on much larger scales than the charge density. It is thus preferable to assume the charge neutrality (this is the so-called {\it{plasma approximation}}) and to compute the electric field by other means.
Examples of such quasi-neutral descriptions can be found, for instance, in the Langmuir model \cite{fernsler_quasineutral_2005,langmuir_interaction_1929}, the electric field being given by the Boltzmann approximation, or in the Magneto-Hydro-Dynamic (MHD) models \cite{biskamp_nonlinear_1997} and the kinetic quasi-neutral models described in \cite{hewett_low-frequency_1994,Joyce1997,acheritogaray_kinetic_2011}, for which the electric field is provided by an Ohm law. As outlined in \cite[chapter 3]{Chen}, ``{\it{in a plasma it is possible to assume $n_i=n_e$ and $\nabla \cdot E \neq 0$ at the same time}}'', leading to the following guideline: ``{\it{Do not use Poisson's equation to obtain $E$ unless it is unavoidable !}}'' which defines the path followed in the present work.



In many plasma problems, the non-neutral areas are very localized but crucial for the global evolution of the plasma. Furthermore, the extent of these regions evolves with time. That is the case, for instance, of Plasma Opening Switches (POS) \cite{zeng_particle--cell_2001,yulan_particle--cell_2003,richardson_particle--cell_2012} where an electromagnetic wave interacts with a dense plasma, leading to the formation of non-neutral sheaths. Quasi-neutral models are not appropriate to describe all the complex physics that occurs in these non-neutral regions. Non-neutral models, such as the Vlasov-Maxwell model or the Euler-Maxwell model, are not appropriate either because the quasi-neutral limit is a singular limit (some equations degenerate in the quasi-neutral limit). In particular, explicit discretizations of non-neutral models are subject to stability conditions on the mesh size and the time step that are all the more restrictive than the Debye length and the plasma period are small.

A first way to cope with these multi-physics problems is to split the domain into non-neutral and quasi-neutral areas, and to use suitable models in each area (see \cite{schumer_mhd--pic_2001} for example). Another way is to develop an Asymptotic-Preserving (AP) method.  Introduced by S. Jin \cite{Jin99} for
transport in diffusive regimes, the Asymptotic-Preserving methods are specifically designed for singular perturbation problems. For a general presentation of the Asymptotic-Preserving methodology and examples of application, we refer to \cite{degond_asymptotic-preserving_2013}. In the case of plasma descriptions and their quasi-neutral limit, an Asymptotic-Preserving discretization should satisfy the following three properties.
\begin{enumerate}[P1.]
\item It is consistent with a non-neutral model.
\item It remains stable in the quasi-neutral limit.
\item It is consistent with a quasi-neutral model in the quasi-neutral limit.
\end{enumerate}
These properties make the Asymptotic-Preserving methods able to treat efficiently both quasi-neutral and non-neutral areas of a plasma. When the mesh size resolves the Debye length, the Asymptotic-Preserving  methods provide standard discretizations of a non-neutral model; conversely when the Debye length is not resolved, they afford discretizations of a quasi-neutral model. Asymptotic-Preserving discretizations have already been developed for various non-neutral models: the  Euler-Poisson \cite{CDV07}, Euler-Maxwell \cite{DDSprep}, Vlasov-Poisson \cite{BCDS09,DDNSV10} and BGK-Vlasov-Poisson \cite{CDV} models. The purpose of the present article is to carry on the work initiated in these former realizations and to derive Asymptotic-Preserving Particle-In-Cell methods for kinetic description of magnetized plasmas, hence the Vlasov-Maxwell model.

The first step in the design of the Asymptotic-Preserving schemes consists in specifying the quasi-neutral model. With this aim, the Vlasov-Maxwell equations are scaled so that they depend on a single dimensionless parameter $\lambda$ quantifying how close to quasi-neutrality the plasma is. This parameter $\lambda$ represents different diemensionless parameters including the ratio of the Debye length to the space scale of interest and the ratio of the plasma period to the time scale of interest. The scaling assumptions are actually similar to the assumptions used to derive the MHD models, the kinetic description and the finite inertia of the particles being however preserved. In the quasi-neutral model obtained with this scaling, the electric field is computed by means of a generalized Ohm law.

The quasi-neutral model is a singular limit of the Vlasov-Maxwell model. Consequently, a straightforward discretization of the Vlasov-Maxwell model would not be consistent with the quasi-neutral model in the quasi-neutral limit (condition P3). The key idea is to reformulate the Vlasov-Maxwell system in a set of equivalent equations for which the quasi-neutral model is a regular limit. The reformulated system unifies the two models, with a smooth transition from one to another according to the value of the asymptotic parameter $\lambda$. By discretizing the reformulated equations, we would obtain schemes that satisfy automatically conditions P1 and P3. However, it is preferable to start from the standard Vlasov-Maxell equations, which ensures the condition P1, and to use the reformulated equations as a guideline to obtain schemes that meet also the condition P3. This approach allows us to use discretizations with good and well-known properties. The reformulated equations show that the source term of the Maxwell-Amp\`ere equation must be predicted using a discretization of the generalized Ohm law in which the electric field is made implicit. We consider two different discretizations of the generalized Ohm law: the first one, namely the AP-Moment discretization, is an Eulerian approximation; the other, the AP-Particle discretization, relies on a partial Lagrangian approximation (using an advance of the particles). Coupled with a suitable particle pusher, these discretizations prove to be stable in the quasi-neutral limit (condition P2).

Particle-In-Cell (PIC) methods are well documented to lack consistency with the Gauss law, which can be the source of non-physical results in the numerical simulations \cite{Bar,Lan92,BCS07}. There are several ways of enforcing the Gauss law in Particle-In-Cell methods (see \cite{Bar} for a thorough review). In the present work, we focus on the elliptic correction of the electric field, which is a simple and robust procedure, and we adapt it to the Asymptotic-Preserving framework.

The Asymptotic-Preserving methods are semi-implicit and share some similarities with the Direct Implicit \cite{LCF83,CLF82,CLHP89,HeLa87} and Implicit Moment methods \cite{Mas81,BrFo82,WBF86,Mas87,RLB02}. However, the aim of these methods and the way they are derived are quite different. These methods are designed to be free from the usual stability constraints on the explicit methods, in order to study large-scale phenomena, but without the aim to ensure the consistency with a well-identified asymptotic model.

The organization of the paper is the following. In Section~\ref{sec:eq:cont}, we scale the Vlasov-Maxwell system, then identify the quasi-neutral model and finally reformulate the Vlasov-Maxwell in a set of equations for which the quasi-neutral model is a regular limit. These operations are performed first on the standard Vlasov-Maxwell system, then on the Vlasov-Maxwell system with elliptic correction. Section~\ref{sec:numeric} is devoted to the derivation of the AP-Moment and AP-Particle schemes. Reference explicit schemes are also presented to illustrate the defects of the standard discretizations of the Vlasov-Maxwell equations. In Section \ref{eq:survey1}, without the purpose to be exhaustive, this topic being the subject of an active research for many years, the AP schemes are compared with other semi-implicit or implicit PIC methods for the Vlasov-Maxwell equations. In Section \ref{sec:static}, the AP schemes are compared with the AP schemes developed for the Vlasov-Poisson equations in \cite{DDNSV10}. Finally, in Section~\ref{sec:simu}, the AP schemes are tested on various and demanding simulations: the classical Landau damping; the expansion of a plasma slab into vacuum; a one-dimensional model of POS; the propagation of a KMC wave in a two-dimensional model of POS.



\section{The Vlasov-Maxwell system and its quasi-neutral limit}\label{sec:eq:cont}

\subsection{The Vlasov-Maxwell system}

For simplicity, the ions are supposed to form a motionless and uniform background density, denoted by $n_i$. The electron evolution is described using a distribution function $f$ depending on the space variable $x\in \Omega_x \subset  \Rset^3$, the microscopic velocity $v \in \Omega_v \subset \Rset^3$ and the time $t \in \Rset^+$. The electron density $n$, the electrical charge and current densities, $\rho$ and $J$, as well as the stress tensor $S$ are defined from the distribution function by
\begin{equation*}
n=\int_{\Omega_v}f(x,v,t)\ dv\,,\quad \rho=e(n_i-n)\,,\quad J=-e\int_{\Omega_v}f(x,v,t)v\ dv\,, \quad S=\int_{\Omega_v}f(x,v,t)v \otimes v\ dv\,,
\end{equation*}
$e$ denoting the elementary charge.
The distribution function $f$ satisfies the Vlasov equation
\begin{equation}\label{eq:std:vlasov}
\partial_t f + v \cdot \nabla_x f -\frac{e}{m}(E+v \times B)\cdot \nabla_v f =0\,,
\end{equation}
where $m$ is the electron mass, $E$ the electric field, and $B$ the magnetic field. The electric and magnetic fields are created by the particles (self-consistent fields) and satisfy the following Maxwell equations: 
\begin{align}
& \frac{1}{c^2}\partial_t E - \nabla \times B= - \mu_0 J\,,\label{eq:std:Amp}\\
& \partial_t B+ \nabla \times E = 0\,,\label{eq:std:Far}\\
&\nabla \cdot E= \frac{\rho}{\epsilon_0}\,,\label{eq:std:divE}\\
& \nabla \cdot B = 0\,,\label{eq:std:divB}
\end{align}
where $c$ is the speed of light, $\mu_0$ the vacuum permeability and $\epsilon_0$ the vacuum permittivity. The above equations \eqref{eq:std:vlasov}-\eqref{eq:std:divB} form the so-called Vlasov-Maxwell system. Of course, this system must be supplemented with initial and boundary conditions specific to each problem. Several examples of problems with their initial and boundary conditions are presented in Section \ref{sec:simu}, devoted to numerical simulations.

In the Vlasov-Maxwell system, the Maxwell-Gauss equation \eqref{eq:std:divE} (or Gauss law) and the Maxwell-Thomson equation \eqref{eq:std:divB} are actually consequences of the other three equations. The integration of the Vlasov equation \eqref{eq:std:vlasov} over $\Omega_v$ gives the continuity equation
\begin{equation}
  \partial_t \rho + \nabla \cdot J = 0\,,
\end{equation}
which translates the conservation of the electron and ion densities. Combining this continuity equation with the divergence of the Maxwell-Amp\`ere equation \eqref{eq:std:Amp}, we obtain
 \begin{equation}
   \partial_t \left(\nabla \cdot E -\frac{\rho}{\epsilon_0}\right)=0\,.
 \end{equation}
 Taking the divergence of the Maxwell-Faraday equation \eqref{eq:std:Far}, we find
\begin{equation}
 \partial_t \nabla \cdot B =0\,.
\end{equation}
The above two equations show that the Gauss law and Maxwell-Thomson equation hold true for $t>0$ as soon as they hold true at $t=0$.

\subsection{Scaling of the Vlasov-Maxwell system}\label{sec:scaling}
\label{sec-scaling}

Consider a plasma characterized by the following dimensional parameters: $x_0$ for the space scale, $t_0$ for the time scale, $n_0$ (taken equal to $n_i$) for the density, $v_0$ (taken equal to $x_0/t_0$) for the velocity, $v_{th,0}$ for the electron thermal velocity and $(E_0,B_0)$ for the electromagnetic field. The Debye length $\lambda_D$ (the typical length of charge separation) and the (electron) plasma period $\tau_p$ (the typical period of electron oscillations) are defined by:
\begin{equation*}
\lambda_D=\sqrt{\frac{m \epsilon_0 v_{th,0}^2}{e^2 n_0}}\,, \qquad \tau_p=\sqrt{\frac{m \epsilon_0}{e^2 n_0}}\,.
\end{equation*}
As explained in the introduction, a plasma is considered to be quasi-neutral when the Debye length is very small compared to the space scale at which the plasma is observed. Thus, the parameter
\begin{equation*}
\lambda = \frac{\lambda_D}{x_0}
\end{equation*}
quantifies how close to quasi-neutrality the plasma is. 

To investigate the quasi-neutral limit ($\lambda \to 0$), we scale the Vlasov-Maxwell system so that it depends only on the parameter $\lambda$. Using the dimensionless variables
\begin{gather*}
  x^* = \frac{x}{x_0} \,, \quad t^* = \frac{t}{t_0} \,,\quad v^* = \frac{v}{v_0} \,, \quad f^* = \frac{f}{n_0/v_0} \,, \quad n^* = \frac{n}{n_0} \,,\quad J^* = \frac{J}{e n_0 v_0} \,, \\ E^* = \frac{E}{E_0} \,,\quad B^* = \frac{B}{B_0} \,,
\end{gather*}
the Vlasov-Maxwell system can be rewritten as (dropping the stars for the sake of readability): 
\begin{align}
&\partial_t f +  \frac{1}{M} v \cdot \nabla_x f - \eta M (E+ \frac{\beta^2}{M} v \times B)\cdot \nabla_v f =0\,,\label{VMadim:a}\\
& \lambda^2 \eta M^2 (\alpha^2\frac{\partial {E}}{\partial t} - \beta^2 \nabla \times B)=-\alpha^2 J\,,\\
& \beta^2 \partial_t B + \nabla \times E = 0\,,\label{VMadim:c}\\
&\lambda^2 \eta M^2 \nabla \cdot E= 1-n\,,\\
& \nabla \cdot B = 0\,,\label{VMadim:e}
\end{align}
where
\begin{equation*}
M = \frac{v_0}{v_{th,0}}\,, \quad \eta=\frac{e x_0 E_0}{mv_0^2}\,, \quad \alpha=\frac{v_0}{c}\,, \quad \beta=\sqrt{\frac{v_0B_0}{E_0}}\,.
\end{equation*}

The parameter $M$ is the Mach number, $\eta$ the ratio of the electric energy to the drift energy, $\alpha$ the ratio of the typical velocity to the speed of light and $\beta$ the induction electric field to the typical electric field. The scaling relations defining the quasi-neutral regime are very similar to the most common assumptions of MHD models: $\lambda \ll 1$, $\alpha \ll 1$, $M=1$, $\eta=1$, and $\beta=1$. A vanishing dimensionless Debye length $\lambda \to 0$ provides the quasi-neutrality assumption thanks to the Gauss law. The choice $M=\eta=1$ means that the drift energy, the thermal energy and the electric energy remain the same order of magnitude. Furthermore, this choice implies $\tau_p/t_0=\lambda$, which means that the parameter $\lambda$ controls the smallness of the plasma period with respect to the time scale of the problem. Therefore, this quasi-neutral regime is a low frequency asymptotic, which explains the stability issues on the time step encountered by standard explicit discretizations.
The identity $\beta=1$ is related to the so-called ``frozen-field'' assumption, translating the property that, in a dense plasma, the magnetic field is convected with the plasma flow. The quasi-neutral regime is also inconsistent with the propagation of electromagnetic waves at the speed of light, hence $\alpha \ll 1$. This defines another small scale, which is identified to $\lambda$ in the sequel, so that the re-scaled system becomes:
%
\begin{align}
&\partial_t f + v \cdot \nabla_x f - (E+ v \times B)\cdot \nabla_v f =0\,,\label{SVM1}\\
& \lambda^2 \partial_t E - \nabla \times B=-J\,, \label{SVM2}\\
& \partial_t B + \nabla \times E = 0\,, \label{SVM3}\\
&\lambda^2 \nabla \cdot E = 1-n\,,\label{SVM4}\\
& \nabla \cdot B = 0\,.\label{SVM5}
\end{align}

\subsection{The quasi-neutral model}\label{sec:QN:regime}

The scaling of the Vlasov-Maxwell system allows us to identify (formally) the quasi-neutral limit of the Vlasov-Maxwell system. Setting $\lambda=0$ in \eqref{SVM1}-\eqref{SVM5}, we obtain
\begin{align}
&\partial_t f + v \cdot \nabla_x f -(E+v \times B)\cdot \nabla_v f =0\,,\label{QNVM1}\\
& \nabla \times B=J\,,\label{QNVM2}\\
& \partial_t B + \nabla \times E = 0\,,\label{QNVM3}\\
&n= 1\,,\label{QNVM4}\\
& \nabla \cdot B = 0\,.\label{QNVM5}
\end{align}

The singular nature of the quasi-neutral limit appears clearly: both the Maxwell-Amp\`ere equation \eqref{QNVM2} and the Gauss law \eqref{QNVM4} degenerate. The electric field $E$ can no longer be computed explicitly with the degenerate Maxwell-Amp\`ere equation \eqref{QNVM2}. Taking the time derivative of \eqref{QNVM2} together with the curl of \eqref{QNVM3}, we obtain the identity
  \begin{equation}\label{eq:ampere:farraday:0}
    \nabla \times \nabla \times E  = - \partial_t J \,,
  \end{equation}
which provides a means of computing the solenoidal part of the electric field. But the Gauss law, which is degenerate, does not allow us to compute the irrotational part of the electric field. 

While the evolution of the electric field is governed by the displacement current in the Maxwell regime (i.e. when the sources do not dominate), it is given by the particles current in the quasi neutral limit. In this latter regime, an auxiliary equation is required to derive an Ohm law and explicit the dependence of the right hand side of \eqref{eq:ampere:farraday:0} with respect to the electric field. It is thus necessary to use moments of the Vlasov equation (as in \cite{fernsler_quasineutral_2005,CDV07,DDSprep,BCDS09,DDNSV10,HK11}). Multiplying the Vlasov equation by $v$, then integrating over $\Omega_v$, we find the following relation on the current density, which can be viewed as a generalized Ohm law,
\begin{equation}\label{momentJ}
   \partial_t J = \nabla \cdot S + n E - J \times B \,.
\end{equation}
Using this relation together with \eqref{eq:ampere:farraday}, we obtain an equation that determines entirely $E$:
\begin{equation}\label{eq:ampere:QN}
   E + \nabla \times \nabla \times E = J\times B - \nabla \cdot S .
\end{equation}

Finally, provided that $\nabla \times B=J$ at the initial time, the quasi-neutral model can be (formally) rewritten as
\begin{align}
&\partial_t f + v \cdot \nabla_x f -(E+v \times B)\cdot \nabla_v f =0\,,\label{QN1bis}\\
& E  + \nabla \times \nabla \times E = J\times B - \nabla \cdot S\,,\label{QN2bis}\\\
& \partial_t B + \nabla \times E = 0\,,\label{QN3bis}\\
& n = 1\,,\label{QN4bis}\\
& \nabla \cdot B = 0\,.\label{QN5bis}
\end{align}

\begin{remark} The rigorous derivation of the quasi-neutral limit for the Vlasov-Maxwell system, i.e. the convergence of the solutions of the Vlasov-Maxwell system (\ref{SVM1})-(\ref{SVM5}) to a solution of the quasi-neutral Vlasov-Maxwell system \eqref{QNVM1}-\eqref{QNVM5} when $\lambda \rightarrow 0$, is an open problem. Even for the simpler Vlasov-Poisson system, few results exist; see for instance the introduction of \cite{HK11} for a review. 
\end{remark}

\subsection{Reformulation of the Vlasov-Maxwell system}\label{sec-reform}
The purpose of this section is to manufacture a set of equations, equivalent to the original Vlasov-Maxwell system \eqref{SVM1}-\eqref{SVM5}, for which the quasi-neutral limit is regular. In other words, the quasi-neutral model \eqref{QN1bis}-\eqref{QN5bis} must be recovered when $\lambda$ is set to zero in this new set of equations. To obtain such equations, we reproduce on \eqref{SVM1}-\eqref{SVM5} the operations performed in the previous section on \eqref{QNVM1}-\eqref{QNVM5}.
%
%
%

The time derivative of \eqref{SVM2} together with the curl of \eqref{SVM3} yield
\begin{equation}\label{eq:ampere:farraday}
\lambda^2 \partial_t^2  E +  \nabla \times \nabla \times E  = - \partial_t J \,.
\end{equation}
Then, combining the above relation with the generalized Ohm law \eqref{momentJ}, we find
\begin{equation}\label{eq:ampere:QN}
\lambda^2 \partial_t^2  E   + n E + \nabla \times (\nabla \times E)= J \times B - \nabla \cdot S\,.
\end{equation}
Finally, the reformulated Vlasov-Maxwell system is
\begin{align}
&\partial_t f + v \cdot \nabla_x f -( E+v \times B)\cdot \nabla_v f =0\label{SVMbis1}\\
&\lambda^2 \partial_t^2  E   + n E + \nabla \times (\nabla \times E)= J \times B - \nabla \cdot S,\label{SVMbis2}\\
& \partial_t B + \nabla \times E = 0,\\
&\lambda^2 \nabla \cdot E = 1-n\,,\label{SVMbis4}\\
& \nabla \cdot B = 0\,. \label{SVMbis5}
\end{align}
This system is equivalent to the Vlasov-Maxwell system \eqref{SVM1}-\eqref{SVM5} provided that the Maxwell-Amp\`ere equation \eqref{SVM2} is satisfied at the initial time.
\subsection{Enforcement of the Gauss law}\label{sec-enfGauss}

In standard Particle-In-Cell methods, the source terms of the Maxwell equations, namely the discrete charge and current densities, are computed using a particle-to-grid assignment. The grid quantities do not satisfy the discrete equivalent of the continuity equation and thus the Gauss law is not enforced, which can be the source of non-physical results in the numerical simulations, as pointed out in \cite{BiLa04,Bar}. Different procedures are successfully used to correct this deficiency; we refer to \cite{Bar} for a thorough review. Two main approaches can be identified. The first one consists in computing a correction of the electric field, this correction being the solution of an elliptic equation (or, in some variants, a parabolic or hyperbolic equation). The second approach modifies the particle-to-grid assignement so that the charge and current densities satisfy the discrete continuity equation. It is interesting to focus on the first approach, and more specifically on the elliptic correction, because the Gauss law, used to compute the correction, degenerates in the quasi-neutral limit. Therefore an adaptation of the elliptic corection needs to be provided.

A rigorous way to add the elliptic correction  is to consider a generalized formulation of the Maxwell equations in which the Gauss law is explicitly enforced by means of a correction field (a Lagrange multiplier) \cite{BCS07}. The generalized formulation that we use here, also called Boris correction \cite{Bor72}, is slightly different, though equivalent, to the one proposed in \cite{BCS07}. This formulation involves an electric field $\tilde{E}$ which does not satisfy the Gauss law, a corrected electric field $E$ which does satisfy the Gauss law and a correction field $p$:
\begin{align}
& \partial_t f + v \cdot \nabla_x f -( E+v \times B)\cdot \nabla_v f =0\,,\label{VMc1}\\
& \lambda^2 \partial_t \tilde{E} - \nabla \times B=-J\,,\label{VMc2}\\
& \partial_t B + \nabla \times E = 0\,,\label{VMc3}\\
&\lambda^2 \nabla \cdot E = 1-n\,,\label{VMc4}\\
& \nabla \cdot B = 0\,,\label{VMc5}\\
& E = \tilde{E} - \nabla p\,.\label{VMc6}
\end{align}
The boundary conditions prescribed on $p$ are chosen to be compatible with those on $E$. The generalized formulation is well-posed even if the continuity equation is not satisfied at each time $t$ and is equivalent to the standard Maxwell system as soon as the continuity equation is satisfied at each time $t$. Combining \eqref{VMc4} and \eqref{VMc6}, we deduce that $p$ is the solution of the elliptic equation 
\begin{equation}\label{VMc7}
 \lambda^2 \Delta p = \lambda^2 \nabla \cdot \tilde{E} - (1-n) \,.
\end{equation}
When the continuity equation holds true at each time $t$, the generalized equations \eqref{VMc1}-\eqref{VMc6} are equivalent to the standard equations \eqref{SVM1}-\eqref{SVM5}. Indeed, in this case, the electric field $\tilde{E}$  satisfies the Gauss law $\lambda^2 \nabla \cdot \tilde{E} =1-n$. Therefore, $p$ vanishes and $E=\tilde{E}$. 

At the time-discretized level, the generalized formulation provides us with an obvious way to compute an electric field satisfying the Gauss law. First, an electric field is computed using \eqref{VMc2}-\eqref{VMc3}; then, this electric field is corrected using \eqref{VMc6} and \eqref{VMc7}. 

If we add the elliptic correction to the quasi-neutral model, we obtain
\begin{align*}
&\tilde{E} + \nabla \times (\nabla \times \tilde{E})= J \times B - \nabla \cdot S\,,\\
& \partial_t B + \nabla \times E = 0\,,\\
&n=1\,,\\
& \nabla \cdot B = 0\,,\\
& E = \tilde{E} - \nabla p\,.
\end{align*}
The degeneracy of the Gauss law prevents from computing $p$ directly. To obtain an equation on $p$, we need again to use the moments of the Vlasov equation. However, the  continuity equation is assumed to be not exactly satisfied by the moments of the distribution function. Following the spirit of the Boris correction, this unconsistency is corrected in the equation by means of an electrostatic deviation of the electric field, giving rise to the following modified continuity equation
\begin{align}
& \partial_t J = \nabla \cdot S + n ( \tilde{E} - \nabla p) - J \times B \,,\label{moment2}\\
& \partial_t n =  \nabla \cdot J \label{moment3}\,,
\end{align}
Introducing these definitions into the time derivative of the Gauss law yields to
\begin{equation*}
  - \lambda^2 \partial_t^2 \Delta p  - \nabla \cdot (n \nabla p ) = - \lambda^2 \partial_t^2 \nabla \cdot \tilde{E} - \nabla^2:S - \nabla \cdot (n \tilde E  ) + \nabla \cdot (J \times B) \,,
\end{equation*}
which is finally equivalent, assuming that the Gauss law is satisfied at initial times, to
\begin{equation}\label{eq:comput:p}
  - \lambda^2 \partial_t^2 \Delta p  - \nabla \cdot (n \nabla p ) = \partial_t^2 n - \nabla^2:S - \nabla \cdot (n \tilde E  ) + \nabla \cdot (J \times B) \,.
\end{equation}

\begin{remark}\label{remark:corrector}
The source term of this equation is the difference beetween the time evolution of the density produced by the Vlasov equation and that of the electric field divergence predicted thanks to the moments. This point will be clarified with the time semi-discretization introduced in the sequel (see section~\ref{sec:general:framework}). The equation \eqref{eq:comput:p} provides a means of preserving the consistency with the Gauss law whatever the values of $\lambda$ and is therefore AP in the quasi-neutral limit.  
\end{remark}

The quasi-neutral model with correction can thus be stated as follows
\begin{align}
&\partial_t f + v \cdot \nabla_x f -( E+v \times B)\cdot \nabla_v f =0\,,\label{QNcbis1}\\
&\tilde{E} + \nabla \times \nabla \times \tilde{E}= J \times B - \nabla \cdot S\,,\label{QNcbis2}\\
&\partial_t B + \nabla \times E = 0\,,\\
& \Delta p = - \partial_t^2 n + \nabla^2 : S + \nabla \cdot \tilde{E}-\nabla \cdot(J\times B)\,,\label{QNcbis4}\\
& \nabla \cdot B = 0\,,\label{QNcbis5}\\
& E = \tilde{E} - \nabla p \,.\label{QNcbis6}
\end{align}
Similarly, the reformulated Vlasov-Maxwell system with correction can be rewritten as
\begin{align}
&\partial_t f + v \cdot \nabla_x f -( E+v \times B)\cdot \nabla_v f =0\,,\label{RVMc1}\\
&\lambda^2 \partial_t^2  \tilde{E}   + n \tilde{E} + \nabla \times \nabla \times \tilde{E}= J \times B - \nabla \cdot S,\label{RVMc2}\\
& \partial_t B + \nabla \times E = 0,\label{RVMc2}\\
&\lambda^2 \partial_t^2 \Delta p + \nabla \cdot ( n \nabla p) =  -\partial_t^2n+\nabla^2 : S + \nabla \cdot (n \tilde{E})- \nabla \cdot(J\times B) \,,\label{RVMc4}\\
& \nabla \cdot B = 0\,,\label{RVMc5}\\
& E = \tilde{E} - \nabla p\,.\label{RVMc6}
\end{align}

\section{Asymptotic-Preserving schemes}\label{sec:numeric}

Now that the Vlasov-Maxwell system has been scaled and the quasi-neutral model has been identified, we can state rigorously the properties that an Asymptotic-Preserving discretization must satisfy.
\begin{enumerate}[P1.]
\item For $\lambda > 0$, the discretization is consistent with the Vlasov-Maxwell system \eqref{VMc1}-\eqref{VMc6} (or \eqref{RVMc1}-\eqref{RVMc6}).
\item The stability conditions on the time step and the mesh size do not depend on $\lambda$.
\item For $\lambda = 0$, the discretization is consistent with the quasi-neutral model \eqref{QNcbis1}-\eqref{QNcbis6}.
\end{enumerate}

In this section, we derive two Asymptotic-Preserving schemes, called AP-Moment scheme and AP-Particle scheme. They use a Yee finite-difference approximation for the fields \cite{Taf} (with a regular rectilinear grid) and standard assignment-interpolation procedures for the particles, such as the nearest grid point or cloud-in-cell procedures \cite{BiLa04,HoEa88}. Other choices could have been made: finite volumes or finite elements instead of finite differences for instance. The schemes are presented in a three-dimensional spatial setting and, for simplicity, the domain is assumed to be a rectangular parallelepiped with periodic boundary conditions. It is straightforward to derive one-dimensional and two-dimensional versions of the schemes.

We also present two reference explicit schemes. These schemes illustrate the defects of the standard discretizations of the Vlasov-Maxwell equations and will be compared to the Asymptotic-Preserving schemes in the numerical tests. Finally, we derive 

\subsection{Definitions and notation}\label{sec:defnum}

\subsubsection{Discrete fields and discrete vector calculus operators}

We consider different kinds of discrete fields on the grid (which is rectilinear and regular): primal and dual scalar fields, edge scalar field, primal and dual vector fields, primal symmetric second-order tensor field.
\begin{itemize}
\item The values of a primal scalar field are located at the vertices of the cells, while the values of a dual field are located at the centers of the cells. The values of an edge scalar field are located at the center of the edges.
\item The components of a primal vector field are located at the center of the edges: the $x$-, $y$-, and $z$-components are located at the edges oriented in the  $x$-, $y$-, and $z$-direction, respectively. The components of a dual vector field are located at the center of the faces: the $x$-, $y$-, and $z$-components are located at the faces normal to the  $x$-, $y$-, and $z$-direction, respectively.
\item The diagonal components of a primal symmetric second-order tensor field are located at the vertices of the grid. The $xy$-, $xz$- and $yz$-components are located at the center of the faces normal to the $x$-direction, $y$-direction, and $z$-direction, respectively.
\end{itemize}

Discrete differential operators can be defined on the discrete fields defined above by using central finite differences (and assuming periodic boundary conditions). 
\begin{itemize}
\item A discrete curl operator $\nabla_h \times$ is defined for the primal and dual vector fields. When applied to a primal vector field (resp. dual vector field), the discrete curl operator yields a dual vector field (resp. primal vector field). Furthermore, if $F_h$ is a primal vector field and $G_h$ a dual vector field, then
\begin{equation}\label{curlid}
\nabla_h \times F_h \cdot G_h - \nabla_h \times G_h \cdot F_h = 0\,.
\end{equation}
\item A discrete divergence operator $\nabla_h \cdot$ is defined for the primal and dual vector fields. When applied to a primal vector field (resp. a dual vector field), the discrete divergence operator yields a primal scalar (resp. a dual scalar field). If $F_h$ is a primal or a dual field, then
\begin{equation}
\nabla_h \cdot (\nabla_h \times F_h )= 0\,.
\end{equation}
\item A discrete gradient operator $\nabla_h $ is defined for the primal and dual scalar fields. When applied to a primal scalar field (resp. a dual scalar field), the discrete gradient operator yields a primal vector (resp. a dual vector field).
\item A discrete divergence operator $\nabla_h \cdot$ is defined for the primal tensor field. It yields a primal vector field.
\item A discrete cross product $\times_h$ between a primal vector field and a dual vector field is defined. It yields a primal vector field. Such a discrete operator is built using local averages.
\end{itemize}

\subsubsection{Notation}
\begin{itemize}
\item The grid spacings in the $x$-, $y$-, and $z$-direction are denoted by $\Delta x$, $\Delta y$, and $\Delta z$, respectively. Let $h=1/\sqrt{\frac{1}{\Delta x^2} + \frac{1}{\Delta y^2} +\frac{1}{\Delta z^2}}$.
\item The time interval is discretized with a uniform time step $\Delta t$. Let $t^{\gamma}=\gamma \Delta t$, for any $\gamma \in \Rset_+$.
\item The discrete electric and magnetic fields at time $t^\gamma$ are denoted by $E_h^\gamma$ and $B_h^\gamma$. The discrete electric field is a primal vector field while the discrete magnetic field is dual vector field. The discrete correction field, denoted by $p_h$, is a primal scalar field.
\item Let $N$ be the number of particles. The vectors containing the position and the velocity of the particles at time $t^\gamma$ are denoted by $X_N^\gamma$ and $V_N^\gamma$, respectively. The position and velocity vectors of the $j$th particle at time $t^\gamma$ are denoted by $X_{N,j}^\gamma$ and $V_{N,j}^\gamma$, respectively.
\item The value of a field $F_h$ interpolated at the position $X_{N,j}$ is denoted by $F_h(X_{N,j})$. 
\item The discrete electron density accumulated from the particles at position $X_N$ as a primal scalar field (resp. an edge scalar field) is denoted by $n_h(X_N)$ (resp. $\hat{n}_h(X_N)$). The discrete current accumulated from the particles of position $X_N$ and velocity $V_N$ as a primal vector field is denoted by $J_h(X_N,V_N)$. The discrete second-order moment accumulated from the particles of position $X_N$ and velocity $V_N$ as a primal tensor is denoted by $S_h(X_N,V_N)$. 
\end{itemize}

\subsection{Reference explicit schemes}
\label{sec-exp-schemes}

The first reference scheme combines a Boris scheme for advancing the particles and a leap-frog discretization for the Maxwell equations \cite[Chapter 15]{BiLa04}:
\begin{align}
& \frac{X_{N,j}^{m+1}-X_{N,j}^m}{\Delta t}=V_{N,j}^{m+\demi},&& \forall j\in\{1,\dots,N\}\,,\label{scheme-vm-exp1}\\
& \frac{V_{N,j}^{m+\demi}-V_{N,j}^{m-\demi}}{\Delta t}=-E^m_h\left(X_{N,j}^{m}\right)-\frac{V_{N,j}^++V_{N,j}^-}{2} \times B_h^{m}( X_{N,j}^{m}),&& \forall j\in\{1,\dots,N\}\,,\label{scheme-vm-exp2}\\
& \lambda^2 \frac{\tilde{E}^{m+1}_h- {E}^{m}_h}{\Delta t} - \nabla_h \times B^{m+\demi}_h = J_h(X_N^{m+1},V_N^{m+\demi})\,,\label{scheme-vm-exp3}\\
& \frac{B^{m+\demi}_h-B^{m-\demi}_h}{\Delta t} + \nabla_h \times  {E}^{m}_h =0\,,\label{scheme-vm-exp4}\\
& E_h^{m+1}=\tilde{E}_h^{m+1} - \nabla_h p_h\,,\label{Boris1}\\
&\lambda^2 \Delta_h p_h = \lambda^2 \nabla_h \cdot \tilde{E}_h^{m+1}-\left(1-n_h(X_N^{m+1})\right)\,,\label{Boris2}
\end{align}
where
\begin{align}
&V_{N,j}^{+}=V_{N,j}^{m+\demi}-\demi \Delta t E^m_h\left(X_{N,j}^{m}\right),&& \forall j\in\{1,\dots,N\},\\ 
&V_{N,j}^{-}=V_{N,j}^{m-\demi}+\demi\Delta t E^m_h\left(X_{N,j}^{m}\right),&& \forall j\in\{1,\dots,N\},\\ &B_h^{m}=\demi\left(B_h^{m-\demi}+B_h^{m+\demi}\right).\label{scheme-vm-exp6}
\end{align}
The above equations are solved in the order (\ref{scheme-vm-exp4}), (\ref{scheme-vm-exp2}), (\ref{scheme-vm-exp1}), (\ref{scheme-vm-exp3}), \eqref{Boris2}, \eqref{Boris1}, so that the scheme is fully explicit. This scheme is subject to a number of stability conditions that are given below (between brackets, the conditions are expressed with dimensional parameters and without the scaling assumptions). The time step must resolve the plasma period:
\begin{equation}\label{cst-plasma}
\Delta t < 2 \lambda  \qquad \left[\Delta t < 2 \tau_p\right]\,.
\end{equation}
The grid spacing must resolve the Debye length:
\begin{equation}\label{cst-Debye}
h < \zeta \lambda \qquad \left[h < \zeta \lambda_D\right]\,,
\end{equation}
where $\zeta$ is a parameter depending on the assignment-interpolation procedure. Otherwise, aliasing will heat up the plasma. Furthermore, the time step and the grid spacing must satisfy a Courant condition involving the velocity of the electrons,
\begin{equation}\label{cst-particle-velocity}
\frac{\Delta t}{h} < 1\qquad \left[v_{th,0}\frac{\Delta t}{h} < 1 \right],
\end{equation}
and and another one involving the speed of light,
\begin{equation}\label{cst-light}
\lambda\frac{\Delta t}{h} < 1 \qquad \left[c\frac{\Delta t}{h} < 1\right].
\end{equation}
The first three constraints are due to the explicit discretization of the particle motion, while the last constraint is due to the explicit discretization of the Maxwell equations. Note that the constant in the right-hand side of the Courant condition \eqref{cst-light} is specific to the Yee finite-difference discretization \cite{Taf}. Because of the stability conditions (\ref{cst-plasma}) and (\ref{cst-Debye}), the above scheme is unstable in the quasi-neutral limit. Therefore, it does not satisfy the condition P2, and hence is not Asymptotic-Preserving.

In the second reference scheme, the leap-frog discretization of the Maxwell equations \eqref{scheme-vm-exp3}-\eqref{scheme-vm-exp4} is replaced by an implicit $\theta$-scheme (with $\demi \leq \theta \leq 1$) but the sources remain explicit:
\begin{align}
& \lambda^2 \frac{\tilde{E}^{m+1}_h-{E}^{m}_h}{\Delta t} - \nabla_h \times \bar{B}^{m+\theta}_h = - J_h(X_N^{m+1},V_N^{m+\demi})\,,\label{scheme-vm-impexp3}\\
& \frac{B^{m+1}_h-B^{m}_h}{\Delta t} + \nabla_h \times \bar{E}^{m+\theta}_h =0\,,\label{scheme-vm-impexp4}
\end{align}
with 
\begin{equation*}
\bar{E}^{m+\theta}_h = \theta \tilde{E}^{m+1}_h + (1-\theta) {E}^{m}_h \,, \qquad \bar{B}^{m+\theta}_h = \theta B^{m+1}_h + (1-\theta) {B}^{m}_h\,.
\end{equation*}
The $\theta$-scheme is unconditionally stable for  $\demi \leq \theta \leq 1$. Therefore, this second scheme is not subject to the stability condition \eqref{cst-light}. However, it is still subject to the conditions (\ref{cst-plasma}) and (\ref{cst-Debye}), and hence is not Asymptotic-Preserving. The computation of $\tilde{E}^{m+1}_h$ and $B^{m+1}_h$ requires the solution of a linear system, which makes this scheme slightly more costly than the first one. Other implicit schemes than the $\theta$-scheme could have been used; see for instance \cite{AGNR80,Bow01,CDL08}. The properties of the $\theta$-scheme (stability, energy conservation, dispersion) are recalled in \ref{sec-prop-theta}. 

\subsection{Asymptotic-Preserving schemes}

\subsubsection{General framework}\label{sec:general:framework}

To derive AP schemes, it is preferable to start from the standard Vlasov-Maxell equations \eqref{VMc1}-\eqref{VMc6} than from the reformulated equations \eqref{RVMc1}-\eqref{RVMc6}. It allows us to use discretizations with good and well-known properties. The reformulation operations described in Sections \ref{sec-reform} and \ref{sec-enfGauss} are however used as a guideline to obtain schemes consistent with the quasi-neutral model. In this section, we derive the common general structure of two AP schemes, called AP-Moment and AP-Particle schemes. Their specificities are addressed in the next two sections.

1. First, the Maxwell equations are discretized using a $\theta$-scheme (with $\demi \leq \theta \leq 1$):
\begin{align}
& \lambda^2 \frac{\tilde{E}^{m+1}_h-E^m_h}{\Delta t} - \nabla_h \times \bar{B}^{m+\theta}_h = - \tilde{J}_h^{m+1}\,,\label{AP1}\\
& \frac{B^{m+1}_h-B^{m}_h}{\Delta t} + \nabla_h \times \bar{E}^{m+\theta}_h =0\,,\label{AP2}
\end{align}
where
\begin{equation*}
\bar{E}^{m+\theta}_h = \theta \tilde{E}^{m+1}_h + (1-\theta) {E}^{m}_h \,, \qquad \bar{B}^{m+\theta}_h = \theta B^{m+1}_h + (1-\theta) {B}^{m}_h\,.
\end{equation*}
An implicit discretization of the Maxwell equations is needed to avoid a Courant condition similar to \eqref{cst-light}. The current $\tilde{J}_h^{m+1}$ used as source term in the discrete Maxwell-Amp\`ere equation \eqref{AP1} is defined by an approximation of the generalized Ohm law \eqref{momentJ}. It is crucial to make the electric field implicit in this approximation to ensure the consistency with the quasi-neutral model. We consider two types of approximation. The first one, called AP-Moment, is based on an Eulerian integration. The second one, called AP-Particle, relies on a partial Lagrangian approximation (using an advance of the particles). In both approximations, the current $\tilde{J}_h^{m+1}$ can be written in the form
\begin{equation}
\tilde{J}_h^{m+1}= {J}_h^{m+1,\star} +  \Delta t \, \hat{n}_h(X_N^m)  \tilde{E}^{m+1}_h\,,\label{AP3}
\end{equation}
where ${J}_h^{m+1,\star}$ is a quantity depending on the scheme. Note that $\tilde{J}_h^{m+1}$ and ${J}_h^{m+1,\star}$ are primal vector fields.

The equations (\ref{AP1})-(\ref{AP3}) are equivalent to the linear system
\begin{equation}  \label{AhZh}
A^m_h \begin{pmatrix} \tilde{E}^{m+1}_h \\ B^{m+1}_h \end{pmatrix} = \begin{pmatrix}
\displaystyle \frac{\lambda^2}{\Delta t^2}E^{m}_h  + \frac{1-\theta}{\Delta t}\nabla_h \times B^{m}_h -\frac{1}{\Delta t}{J}_h^{m+1,\star} \\
\displaystyle \frac{1}{\Delta t^2}B^{m}_h - \frac{1-\theta}{\Delta t}\nabla_h \times E^{m}_h\\
\end{pmatrix}\,,
\end{equation}
where $A^m_h$ is a linear operator defined by
\begin{equation}\label{Ah}
A^m_h \begin{pmatrix} E_h \\ B_h \end{pmatrix}{:=}\begin{pmatrix}
\displaystyle \Big(\frac{\lambda^2}{\Delta t^2}+ \hat{n}_h(X_N^m) \Big) E_h - \frac{\theta}{\Delta t} \nabla_h \times B_h\\
\displaystyle  \frac{1}{\Delta t^2} B_h +\frac{\theta}{\Delta t} \nabla_h \times E_h \\
\end{pmatrix}\,,
\end{equation}
for any fields $E_h$ and $B_h$.
In practice, the electric field $\tilde{E}^{m+1}_h$ and the magnetic field $B^{m+1}_h$ are computed by solving this linear system. It is invertible and well-conditionned even when $\lambda \ll \Delta t$ and $\lambda \ll h$. Indeed, using the identity (\ref{curlid}), we obtain the inequality
\begin{equation*}
\begin{pmatrix} E_h \\ B_h \end{pmatrix}^T A^m_h \begin{pmatrix} E_h \\ B_h \end{pmatrix}\geq \left(\frac{\lambda^2}{\Delta t^2} + \min \big(\hat{n}_h(X_N^m)\big)\right) \|E_h\|^2 +  \frac{1}{\Delta t^2} \|B_h\|^2.
\end{equation*}
When $\lambda \ll \Delta t$ and $\lambda \ll h$, the density $\hat{n}_h(X_N^m)$ is expected to be close to 1 (the charge separation is negligible at scales far larger than $\lambda$). Therefore, the quantity $(\lambda^2/\Delta t^2) +\min \big(\hat{n}_h(X_N^m)\big)$ is close to 1.

2. In a second step, the electric field is corrected. Equations \eqref{VMc4} and \eqref{VMc6} are discretized as follows:
\begin{align}
& \lambda^2 \nabla_h \cdot E^{m+1}_h = 1 - n^{m+1}_h \,, \label{AP4}\\
& E^{m+1}_h = \tilde{E}^{m+1}_h - \nabla_h p_h\,.\label{AP5}
\end{align}
The charge density $n^{m+1}_h$ used as source term in \eqref{AP4} is built using an approximation of the moment equations \eqref{moment2} and \eqref{moment3}:
\begin{align}
& n_h^{m+1} = {n}_h(X_N^{m}) + \Delta t \nabla_h \cdot J_h^{m+1} \,,\label{AP6}\\
& J_h^{m+1}= J_h^{m+1,\star} + \Delta t \, \hat{n}_h(X_N^m) \; E^{m+1}_h \,.\label{AP7}
\end{align}
In the above equations, $n_h^{m+1}$ is a primal scalar field and $J_h^{m+1}$ is a primal vector field. It is essential to make the electric field implicit in this approximation to guarantee the consistency with the quasi-neutral model.

From \eqref{AP4}-\eqref{AP7}, we deduce that the correction $p_h$ is solution of the equation
\begin{equation}\label{APcorr1}
-\nabla_h \cdot \left(\left(\frac{\lambda^2}{\Delta t ^2} + \hat{n}_h(X_N^{m})\right) \nabla_h p_h\right) = \frac{1-n_h(X_N^{m})}{\Delta t^2} -\nabla_h \cdot \left(\left(\frac{\lambda^2}{\Delta t ^2} + \hat{n}_h(X_N^{m})\right) \tilde{E}_h^{m+1}\right) - \Delta t \nabla_h \cdot J_h^{m+1,\star}\,.
\end{equation}
Just as the linear system \eqref{AhZh}, the above linear system \eqref{APcorr1} is invertible and well-conditionned even when $\lambda \ll \Delta t$ and $\lambda \ll h$. Using the divergence of \eqref{AP1}, it can be rewritten in a simpler form:
\begin{equation}\label{APcorr2}
-\nabla_h \cdot \left(\left(\frac{\lambda^2}{\Delta t ^2} + \hat{n}_h(X_N^{m})\right) \nabla_h p_h\right) = \frac{1-n_h(X_N^{m})}{\Delta t^2} - \frac{\lambda^2}{\Delta t^2}\nabla_h \cdot E_h^{m}.
\end{equation}
Note that the right-hand side of \eqref{APcorr2} evaluates the inconsistency of the Gauss law at time $t^m$ which amounts to the difference beetween the density accumulated from the particles and the electric field implicitly predicted at the previous time step. In practice, the correction field $p_h$ is computed with \eqref{APcorr2}, then the electric field is corrected with \eqref{AP5}.

3. Finally, the particles are advanced with a Boris-like scheme:
\begin{align}
&\frac{X_{N,j}^{m+1}-X_{N,j}^m}{\Delta t}=V_{N,j}^{m+1},&& \forall j\in\{1,\dots,N\},\label{AP8}\\
& \frac{V_{N,j}^{m+1}-V_{N,j}^{m}}{\Delta t}=- E^{m+1}_h\left(X_{N,j}^{m}\right)-\frac{V_{N,j}^++V_{N,j}^-}{2} \times B_h^{m}( X_{N,j}^{m}),&& \forall j\in\{1,\dots,N\},\label{AP9}
\end{align}
where
\begin{align}
&V_{N,j}^{+}=V_{N,j}^{m+1}-\demi \Delta t E^{m+1}_h\left(X_{N,j}^{m}\right),&& \forall j\in\{1,\dots,N\},\\ 
&V_{N,j}^{-}=V_{N,j}^{m}+\demi \Delta t E^{m+1}_h\left(X_{N,j}^{m}\right),&& \forall j\in\{1,\dots,N\}.\label{AP11}
\end{align}
With this particle pusher, the scheme is subject to the Courant condition \eqref{cst-particle-velocity}. However it is not subject to conditions (\ref{cst-plasma}) and \eqref{cst-Debye} and thus satisfies Property P2. Moreover, it presents favorable conservation properties. The choice of the particle pusher is discussed in more detail in \ref{sec:Part:Pusher}.



\subsubsection{AP-Moment scheme}\label{SecAPm}

The AP-Moment Scheme relies on the following approximation of the generalized Ohm law \eqref{momentJ}:
\begin{equation}\label{APm-1}
\frac{\tilde{J}_h^{m+1}-J_h(X_N^m,V_N^m)}{\Delta t} - \nabla_h \cdot S_h(X_N^m,V_N^m) - \hat{n}_h(X_N^m) \tilde{E}^{m+1}_h +  J_h(X_N^m,V_N^m) \times_h B^m_h =0\,.
\end{equation}
This amounts to set
\begin{equation}\label{eq:def:Jstar}
  J_h^{m+1,\star} = J_h(X_N^m,V_N^m) + \Delta t \Big( \nabla_h \cdot S_h(X_N^m,V_N^m) -  J_h(X_N^m,V_N^m) \times_h B^m_h \Big)\,.
\end{equation}

\begin{prop}\label{prop:consistancy}
 The AP-Moment Scheme, defined by \eqref{AP1}--\eqref{AP3}, \eqref{AP5}, \eqref{APcorr2}, \eqref{AP8}--\eqref{AP11} and \eqref{eq:def:Jstar}, is consistent with the Vlasov-Maxwell system when $\lambda > 0$. It is consistent with the quasi-neutral model \eqref{QNcbis1}-\eqref{QNcbis6} when $\lambda = 0$.
\end{prop}
\begin{proof}
For simplicity, we fix $\theta=1$. First, the system \eqref{AhZh}-\eqref{Ah} yields
\begin{multline}\label{eq:proof:consistancy}
  {\frac{\lambda^2}{\Delta t^2} \left( \tilde{E}^{m+1}_h - {E}^{m}_h \right) = \frac{1}{\Delta t}\Big(\nabla_h \times B_h^m - J^m_h(X_N^m,V_N^m) \Big)}\\ - {\nabla_h \times \nabla_h \times {\tilde{E}^{m+1}_h} -\hat{n}_h(X_N^m)\tilde{E}^{m+1}_h - \nabla_h\cdot S^m_h + J^m_h(X_N^m,V_N^m) \times B^m_h \,.}
\end{multline}
The Amp\`ere law is assumed to be satisfied at the precedent time step, with
\begin{equation*}
  \nabla_h \times B^m_h - J^m_h(X_N^m,V_N^m)  \approx \frac{\lambda^2}{\Delta t} \left(E^m_h - E^{m-1}_h \right) \,,
\end{equation*}
so that the following approximation holds
\begin{equation*}
  \frac{\lambda^2}{\Delta t^2} \left(\tilde{E}^{m+1}_h - 2 {E}^{m}_h + E^{m-1}_h\right) + \nabla_h \times \nabla_h \times\tilde{E}^{m+1}_h + \hat{n}_h(X_N^m)\tilde{E}^{m+1}_h + \nabla_h\cdot S^m_h - J^m_h(X_N^m,V_N^m) \times B^m_h \approx 0 \,,
\end{equation*}
defining a consistent discretization of the reformulated Amp\`ere law \eqref{SVMbis1}.

Next, the equations \eqref{AP4}-\eqref{AP7} and \eqref{eq:proof:consistancy} provide 
\begin{multline*}
  -\nabla_h \cdot \Big(\big(\frac{\lambda^2}{\Delta t ^2} +  \hat{n}_h(X_N^{m})\big) \nabla_h p_h\Big) =  \frac{1}{\Delta t^2} \Big( 1 -n_h^{m+1} - \lambda^2 \nabla_h \cdot E^{m}_h \Big) \\ +\nabla_h \cdot  J^m_h(X_N^m,V_N^m) + \nabla_h^2 : S^m - \nabla_h\cdot ( J_h^m \times B^m ) - \nabla_h \cdot  ( \hat{n}_h(X_N^m)\tilde{E}^{m+1}_h ) \,.
\end{multline*}
Assuming that the Gauss law as well as the continuity equations are satisfied at the previous time step so that the following approximations hold
\begin{align*}
  \lambda^2 \nabla_h \cdot E^m_h &\approx 1 -  {n}_h(X_N^{m})\,,\\
  \Delta t\nabla_h \cdot  J^m_h(X_N^m,V_N^m) &\approx  {n}_h(X_N^{m}) - {n}_h(X_N^{m-1})\,, \\
\end{align*}
the following equation can be stated
\begin{multline*}
  -\nabla_h \cdot \Big(\big(\frac{\lambda^2}{\Delta t ^2} +\hat{n}_h(X_N^{m})\big) \nabla_h p_h\Big) \approx \frac{1}{\Delta t^2} \Big( -n_h^{m+1} +2 {n}_h(X_N^{m}) - {n}_h(X_N^{m-1}) \Big) \\- \nabla_h^2 : S^m - \nabla_h\cdot ( J_h^m \times B^m ) - \nabla_h \cdot  ( \hat{n}_h(X_N^m)\tilde{E}^{m+1}_h) \,,
\end{multline*}
which defines a time discretization of the reformulated Gauss law \eqref{RVMc4} provided that the correction at time level $m$ and $m-1$ vanish.

\end{proof}
\subsubsection{AP-Particle scheme}
\label{SecAPp}
In the AP-Particle scheme, the contributions of $J\times B$ and $\nabla \cdot S$ in the generalized Ohm law \eqref{momentJ} are approximated with a particle advance. Observing that 
\begin{equation}
\nabla \cdot S - J\times B =\int_{\Omega_v} \left( v \cdot \nabla_x f - ( v \times B)\cdot \nabla_v f\right)v\ dv \,,
\end{equation}
we set
\begin{equation}\label{eq:def:Jstar:P}
J^{m+1,\star}_h = J_h({X}^{m+1,\star}_N,{V}_N^{m+1,\star}) \,,
\end{equation}
with
\begin{align}
&\frac{{X}_{N,j}^{m+1,\star}-X_{N,j}^m}{\Delta t}={V}_{N,j}^{m+1,\star},&& \forall j\in\{1,\dots,N\},\label{APp-tilde1}\\
& \frac{{V}_{N,j}^{m+1,\star}-V_{N,j}^m}{\Delta t}= - V_{N,j}^{m} \times B_h^{m}(X_{N,j}^{m}),&& \forall j\in\{1,\dots,N\}.\label{APp-tilde2}
\end{align}
Compared to the AP-Moment scheme, the AP-Particle scheme requires an additional push of the particles, from the state $(X_{N}^{m},V_{N}^{m})$ to the state  $(X_{N}^{m+1,\star},V_{N}^{m+1,\star})$. However, it avoids the computation of the tensor $S_h$. 

\begin{prop}
The AP-Particle Scheme, defined by \eqref{AP1}--\eqref{AP3}, \eqref{AP5}, \eqref{APcorr2}, \eqref{AP8}--\eqref{AP11} and \eqref{eq:def:Jstar:P}--\eqref{APp-tilde2}, is consistent with the Vlasov-Maxwell system when $\lambda > 0$. It is consistent with the quasi-neutral model \eqref{QNcbis1}-\eqref{QNcbis6} when $\lambda = 0$.
\end{prop}
\begin{proof}
When the number of particles is sufficiently large, the following appoximation is valid:
\begin{equation}\label{APp-current-bis}
J_h({X}_N^{m+1,\star},{V}_N^{m+1,\star}) \approx J_h(X_N^m,V_N^m)- \Delta t\ \nabla_h \cdot S_h(X_N^m,V_N^m) + \Delta t  J_h(X_N^m,V_N^m) \times B^m_h\,.
\end{equation}
This proves that the AP-Particle scheme shares the same consistency properties than the AP-Moment scheme.
\end{proof}

\section{Comparison with other implicit PIC methods}\label{eq:survey}

\subsection{Fully implicit, Direct Implicit and Implicit Moment methods}\label{eq:survey1}

Some level of implicitness is required in the discretization of the Vlasov-Maxwell equations to obtain AP schemes. It is needed to ensure the consistency with the quasi-neutral model and the stability in the quasi-neutral limit. Other implicit kinetic methods have been developed since the 1980s. Designed for the simulation of large-scale phenomena, they are derived to relax the main stability conditions that explicit methods must satisfy. We refer to \cite{lapenta_particle_2012} for a recent review. In this section, we discuss the AP character of the three main classes of implicit methods (the fully implicit methods, the Implicit Moment methods and the Direct Implicit methods) and compare them with the AP-Moment and AP-Particle schemes.


Theoretically, a fully implicit discretization of the Vlasov-Maxwell system is stable for any discretization parameters. The resulting problem is a huge system of coupled nonlinear equations (the particle equations and the field equations). Impressive realizations have been achieved in the past few years with the use of Jacobian-Free-Newton-Krylov solvers, preconditioning, optimized and often massively parallel implementations \cite{chen_energy-_2011,chen_efficient_2012,chen_fluid_2014,markidis_energy_2011,taitano_development_2013}. Nevertheless, fully implicit discretizations are still too costly for multi-dimensional simulations. In fully implicit discretizations, the electric field at the advanced time level $E_h^{n+1}$ occurs in the discretization of the terms ${\partial E}/{\partial t}$ and $\nabla \times E$ in the Maxwell equations and is also involved in the definition of the source terms, via the particle equations. Therefore the consistency with the quasi-neutral model is recovered when the coupled system is solved.

 
The Implicit Moment methods \cite{Mas81,BrFo82,WBF86,Mas87,VuBr92,RLB02,lapenta_kinetic_2006,markidis_multi-scale_2010} are semi-implicit methods. They decouple the field equations from the particle equations by using macroscopic evolution equations on $\rho$ and $J$ to predict the sources of the field equations at the advanced time level $\rho_h^{n+1}$ and $J_h^{n+1}$. This approach reduces dramatically the computational cost compared to fully implicit methods, while retaining favorable stability properties. The use of moment equations to predict the sources at the advanced time level is a common point with the AP-Moment and AP-Particle schemes. In particular, the discretization of these moment equations is very similar in the Implicit Moment methods and the AP-Moment scheme.

The Direct Implicit methods \cite{CLF82,LCF83,CLHP89,HeLa87} are also semi-implicit. They enjoy the same accuracy and stability properties as the Implicit Moment methods. They rely on a linearization of the fully implicit discretization, which decouple the field equations from the particle equations. The sources of the field equations are linearized around explicitly extrapolated positions of the particles. The source prediction in the Direct Implicit methods can actually be interpreted as an approximation of the moment equations, comparable to the one used in the AP-Particle scheme. The extrapolation step plays essentially the same role as the first particle advance in the AP-Particle scheme (though being more intricate in the Direct Implicit approach).

From this brief review, we can infer that the fully implicit, Implicit Moment and Direct Implicit methods are generally AP in the quasi-neutral limit. The fully implicit methods are far more costly than the AP-Moment and AP-Particle schemes. The Implicit Moment and Direct Implicit methods share some similarities with the AP-Moment and AP-Particle schemes in their formulation. However, their motivation and the methodology used to derive them differ significantly. The aim of the AP-Moment and AP-Particle schemes is to be consistent with a clearly defined quasi-neutral model, not to relax stability conditions. The elimination of the stability conditions related to the plasma period and the speed of light is a consequence of scaling assumptions made in Section~\ref{sec:scaling} for the definition of the quasi-neutral limit (the scaled plasma period and the ratio of the typical velocity to the speed of light vanish in the quasi-neutral limit). The derivation of the AP-Moment and AP-Particle schemes relies on a reformulation of the Vlasov-Maxwell system which unifies the Vlasov-Maxwell model and the quasi-neutral model in a single set of equations. This reformulation highlights the terms that need to be built or implicited in order to ensure consistency with the quasi-neutral model. This methodology allows us to limit the computational cost of the schemes to what is necessary in view of the AP property. Furthermore, this methodology could be applied to quasi-neutral models with a more reduced complexity, leading to more efficient numerical methods for some  kinds of problems.

\medskip

\subsection{AP-Moment and AP-Particle schemes in the electrostatic regime}\label{sec:static}

The electrostatic regime is characterized by a vanishing magnetic field. In the dimensionless system \eqref{VMadim:a}--\eqref{VMadim:e}, this amounts to the asymptotic  $\beta \to 0$. In this regime, the electric field is irrotational, since the Maxwell-Faraday equation \eqref{VMadim:c} simplifies into $\nabla \times E=0$, and it is assumed to derive from a scalar potential (there exists a scalar field $\phi$ such that $E = -\nabla \phi$). As a consequence, the Gauss law is sufficient to determine completely the electric field and the Vlasov-Maxwell system reduces to the Vlasov-Poisson system:
\begin{align}
&\partial_t f + v \cdot \nabla_x f + \nabla \phi \cdot \nabla_v f =0\,,\label{SVP1}\\
&-\lambda^2 \Delta \phi = 1-n\,.\label{SVP2}
\end{align}

Let us examine the AP-Moment and AP-Particle schemes in the electrostatic regime. Setting the magnetic field to zero, they read
\begin{align}
& \Big(\frac{\lambda^2}{\Delta t^2} + \hat{n}_h(X_N^m) \Big) \tilde{E}_h^{m+1} = \frac{\lambda^2}{\Delta t^2}E^{m}_h  -\frac{1}{\Delta t}J_h^{m+1,\star}\,,\label{AP-VP1}\\
& - \nabla_h \cdot \Big(\Big(\frac{\lambda^2}{\Delta t^2} + \hat{n}_h(X_N^m) \Big) \nabla_h p_h\Big) = \frac{1 - {n}_h(X_N^{m})}{\Delta t^2} -\frac{\lambda^2}{\Delta t^2} \nabla_h \cdot E_h^{m} \,,\label{AP-VP2}\\
& E_h^{m+1} = \tilde{E}_h^{m+1} - \nabla_h p_h \,,\label{AP-VP3}\\
& \frac{X_{N,j}^{m+1}-X_{N,j}^m}{\Delta t}=V_{N,j}^{m+1}\,,\qquad \forall j\in\{1,\dots,N\}\,,\label{AP-VP4}\\ 
& \frac{V_{N,j}^{m+1}-V_{N,j}^m}{\Delta t}= -E_h^{m+1}(X_{N,j}^{m})\,,\qquad \forall j\in\{1,\dots,N\}\,,\label{AP-VP5}
\end{align}
with, for the AP-Moment scheme,
\begin{equation}
J_h^{m+1,\star} = J_h(X_N^m,V_N^m) + \Delta t \, \nabla_h \cdot S_h(X_N^m,V_N^m)\,.
\end{equation}
and, for the AP-Particle scheme,
\begin{equation}
J_h^{m+1,\star}=J_h({X}_N^{m}+\Delta t\, V_N^{m},V_N^{m})\,.
\end{equation}


In a one-dimensional spatial setting, the AP-Moment and AP-Particle schemes can be simplified further. Indeed, for any primal vector field $F_h$, there is a primal scalar field $\phi_h$ such that $F_h= - \nabla \phi_h$. Let $\phi_h^{m+1}$ and $\tilde{\phi}^{m+1}_h$ be the potentials associated with $E_h^{m+1}$ and $\tilde{E}_h^{m+1}$, respectively. Equation \eqref{AP-VP3} implies the equality $\nabla_h \phi^{m+1}_h = \nabla_h \left(\tilde{\phi}^{m+1}_h + p_h\right)$. Then, adding \eqref{AP-VP1} and \eqref{AP-VP2}, we deduce
\begin{equation}\label{AP-VP6}
-\nabla_h \cdot \left(\left(\frac{\lambda^2}{\Delta t ^2} + \hat{n}_h(X_N^{m})\right) \nabla_h \phi^{m+1}_h\right) = \frac{1-n_h(X_N^{m})}{\Delta t^2} - \Delta t \nabla_h \cdot J_h^{m+1,\star} \,.
\end{equation}
This equation determines entirely the potential $\phi^{m+1}_h$ and thus the electric field $E^{m+1}_h$. Therefore, the Maxwell-Amp\`ere equation can be disregarded and the schemes reduce to the equations \eqref{AP-VP4}, \eqref{AP-VP5} and \eqref{AP-VP6}. We remark that the AP-Moment scheme is, in this case, equivalent to the PIC-AP2 scheme introduced in \cite{DDNSV10}.

\section{Numerical simulations}\label{sec:simu}

In order to assess their efficiency and investigate their properties, the AP schemes are tested on various problems and compared with reference explicit schemes (those described in Section \ref{sec-exp-schemes}). The first two problems, namely the classical Landau damping and the expansion of a plasma slab into vacuum, are one-dimensional in space and velocity and purely electrostatic. The two other problems, which describe a POS (Plasma Opening Switch), involve magnetized plasmas. The first one is a very simplified model of POS (one-dimensional in space and two-dimensional in velocity); the second one, more realistic, is two-dimensional in space and velocity and shows the propagation of a KMC wave.

The framework of the above problems is not exactly the general framework described in Sections \ref{sec:eq:cont} and \ref{sec:numeric} (the problems are one-dimensional or two-dimensional, the magnetic field is sometimes disregarded, the ions are not always motionless). The numerical schemes are implemented accordingly. These developments being straightforward, we do not detail them. However, for sake of clarity, we state the version of the Vlasov-Maxwell system used for each problem. All the simulations are performed with the cloud-in-cell assignement-interpolation method and the parameter $\theta$ is always taken equal to 1 in the $\theta$-schemes.

\subsection{Landau damping}

When a plasma in a spatially homogeneous equilibrium state is slightly perturbed, it returns exponentially fast, with oscillations, toward its initial equilibrium state. This is the so-called Landau damping. This problem allows us to test the ability of the AP schemes to accurately reproduce phenomena occurring at the Debye length and plasma period scales.

The non-neutral model used for this problem is a one-species, one-dimensional, electrostatic model (the magnetic field is disregarded, the ions are motionless and form a uniform background). Scaled as in Section \ref{sec:scaling}, the equations are
\begin{align}
&\partial_t f + v\,\partial_x f - E\,\partial_{v} f  =0\,,\\
& \lambda^2 \partial_t E = -J\,,\quad \lambda^2 \partial_x E = 1-n\,.
\end{align}
The space domain is $[0,4\pi]$ and the scaled Debye length $\lambda$ is taken equal to 1. The initial electron density follows a Maxwellian distribution with a small spatial perturbation:
\begin{equation}\label{eq:E:landau;analytic}
f_{0}(x,v)=\left(1+\alpha\cos\left(\frac{x}{2}\right)\right)\frac{1}{\sqrt{2 \pi}}e^{-\frac{v^2}{2}},
\end{equation}
where $\alpha = 5\cdot 10^{-2}$. Periodic and homogeneous Dirichlet boundary conditions are prescribed for the particles and the electric field, respectively. An analytical approximation of the electric field can be computed by applying a Laplace-Fourier transform to the linearized system; see \cite{CRS09} for the detailed calculation. Keeping only the dominating mode, the others being quickly damped, the following approximation is obtained:
\begin{equation}\label{Landau-analytic}
E(x,t) \approx -1.4708 \alpha e^{-0.1533 t}\cos(1.4156 t-0.536245)\sin\left(x/2\right).
\end{equation}

\begin{figure}[!ht]
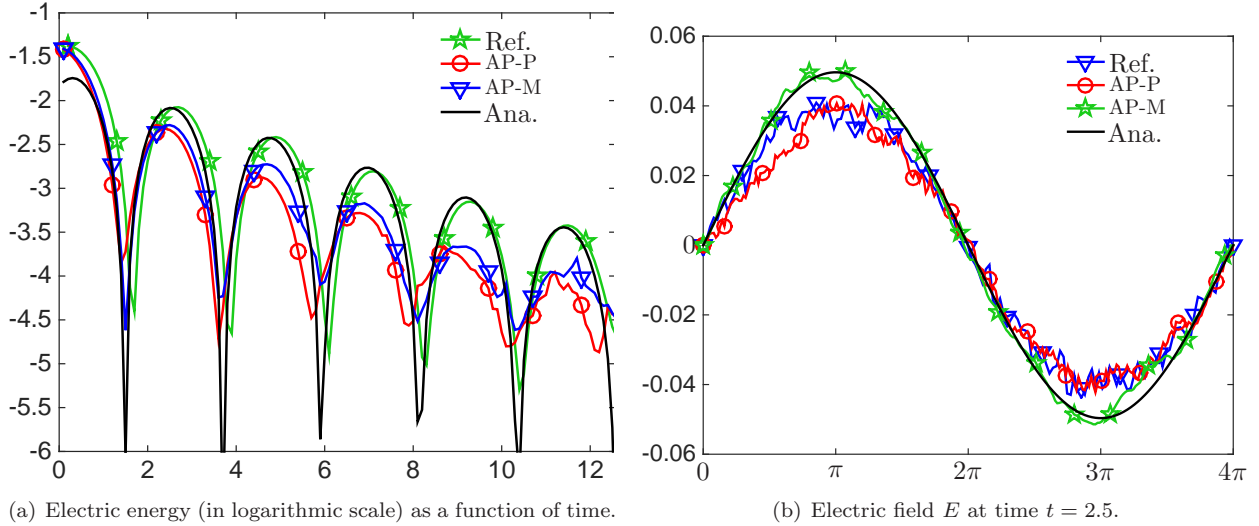

\begin{center}
\psfrag{Explicit}[l][l][1.]{Ref.}
\psfrag{Particle}[l][l][.8]{AP-P}
\psfrag{Moment}[l][l][.8]{AP-M}
\psfrag{Analytique}[l][l][1.]{Ana.}
\subfigure[Electric energy (in logarithmic scale) as a function of time.]{%
\includegraphics[width=0.49\textwidth]{Landau-ExHist.eps}}\hspace*{0.02\textwidth}%
\subfigure[Electric field  $E$ at time $t=2.5$.]{%
\psfrag{0}[][][1.]{$0$}
\psfrag{Pi}[][][1.]{$\pi$}
\psfrag{2Pi}[][][1.]{$2\pi$}
\psfrag{3Pi}[][][1.]{$3\pi$}
\psfrag{4Pi}[][][1.]{$4\pi$}
\includegraphics[width=0.49\textwidth]{Landau-Ex.eps}}
\caption{Landau damping. Analytical approximations (Ana.) and results computed using the reference scheme (Ref.), the AP-Particle scheme (AP-P), the AP-Moment scheme (AP-M). The simulations are performed with $\Delta x \approx \lambda/20$, $\Delta t = \lambda/10$ and $10^6$ particles.}
\label{Landau-E}
\end{center}
\end{figure}

Numerical simulations are performed with discretization parameters smaller than the Debye length and the plasma period. The numerical results, together with the analytical approximations given by the formula \eqref{Landau-analytic}, are represented in Figure~\ref{Landau-E}. We observe that the AP schemes reproduce quite accurately the oscillation period of the Landau damping but they are more dissipative than the reference explicit scheme (the first one described in Section \ref{sec-exp-schemes}). Note that the discrepancy between the numerical results and the analytical approximation at the first oscillation is due to the neglected Laplace-Fourier modes in the analytical approximation. The conservation properties of the AP schemes for the plasma oscillations could probably be improved with a more careful design of the time discretization, with centered or high-order approximations instead of backward approximations.

\subsection{Plasma expansion}

The second problem is the expansion of a plasma slab \cite{GMAH08,DDNSV10,manfredi_vlasov_2011}. The plasma is initially confined in a small area at the center of the domain and surrounded by vacuum. A non-neutral sheath forms at the plasma-vacuum transition and the large electric field created in this sheath accelerates the ions, leading to the plasma expansion. While the quasi-neutral model is able to account for the ion motion and the plasma expansion in the plasma bulk, it is inaccurate in the sheath. Therefore, this problem allows us to verify, on the one hand, the consistency of the AP schemes with the quasi-neutral model and their stability in the quasi-neutral regime and, the other hand, the consistency with the non-neutral plasma description, with the transition from one regime to the other. It is also an excellent test for the energy conservation properties of the schemes since it relies on a kinetic energy transfer from the electrons to the ions via the electric field created in the sheath.

The non-neutral model used for this problem is a two-species, one-dimensional, electrostatic model (the magnetic field is disregarded, the ions are not motionless). This test case is stated and implemented in dimensional variables, the equations being:
\begin{align*}
&\partial_t f_e + v\, \partial_x f_e - \frac{e}{m_e}E\, \partial_{v} f_e  =0,\\
&\partial_t f_i + v\, \partial_x f_i + \frac{e}{m_i}E\, \partial_{v} f_i =0,\\
& \partial_t E = -\frac{J}{\epsilon_0}\,, \quad \partial_x E =  \frac{\rho}{\epsilon_0} \,,
\end{align*}
where the indices $i$ and $e$ denote the quantities related to the ions and the electrons, respectively.

The space domain is $[-L,L]$ with $L= 1$ m. The problem being symmetric, the simulations are actually performed only on the half domain $[0,L]$. The ion mass is such that $m_i=1836\, m_e$. The initial ion density is equal to $n_{i0}$ for $ x\in [0,0.02]$ and is zero in the rest of the domain (different values of $n_{i0}$ are used in the simulations). The initial electron density $n_{e0}$ satisfies the Maxwell-Boltzmann relation $n_{e0}=n_{i0} \exp(\phi_0)$, where $\phi_0$ is the electrostatic potential, solution of $-\Delta \phi_0 = e(n_{i0}-n_{e0})/\epsilon_0$ (this nonlinear problem is solved numerically). The initial electron and ion velocities follow Maxwellian distributions with zero mean velocities and respective temperatures $T_e$ and $T_i$. The electron temperature $T_e$ is chosen so that the initial electron thermal velocity is equal to $v_{th,e}=1$ m$\cdot$s$^{-1}$ and the ion temperature is such that $T_i = 10^{-3}T_e$. To represent the symmetry at the left end of the half domain, a specular reflection condition is prescribed for the distribution functions (i.e. exiting particles are reinjected with reversed velocities) and an homogeneous Dirichlet condition is enforced on the electric field. At the right end, an absorbing condition is enforced on the distribution functions (i.e. exiting particles are not reinjected) and an homegeneous Neumann condition is enforced on the electric field. In addition to the time and space scales related to the electron motion, this 
problem involves time scales related to the ion motion, namely the ion plasma period $\tau_{pi}=\sqrt{m_i \epsilon_0/(e^2 n_{i0})}$ and the ion accoustic wave speed $c_s=\sqrt{k_B T_e/m_i}$. The speed of the plasma expansion is theoretically the same order as $c_s$ \cite{GMAH08}.

\begin{figure}[!ht]
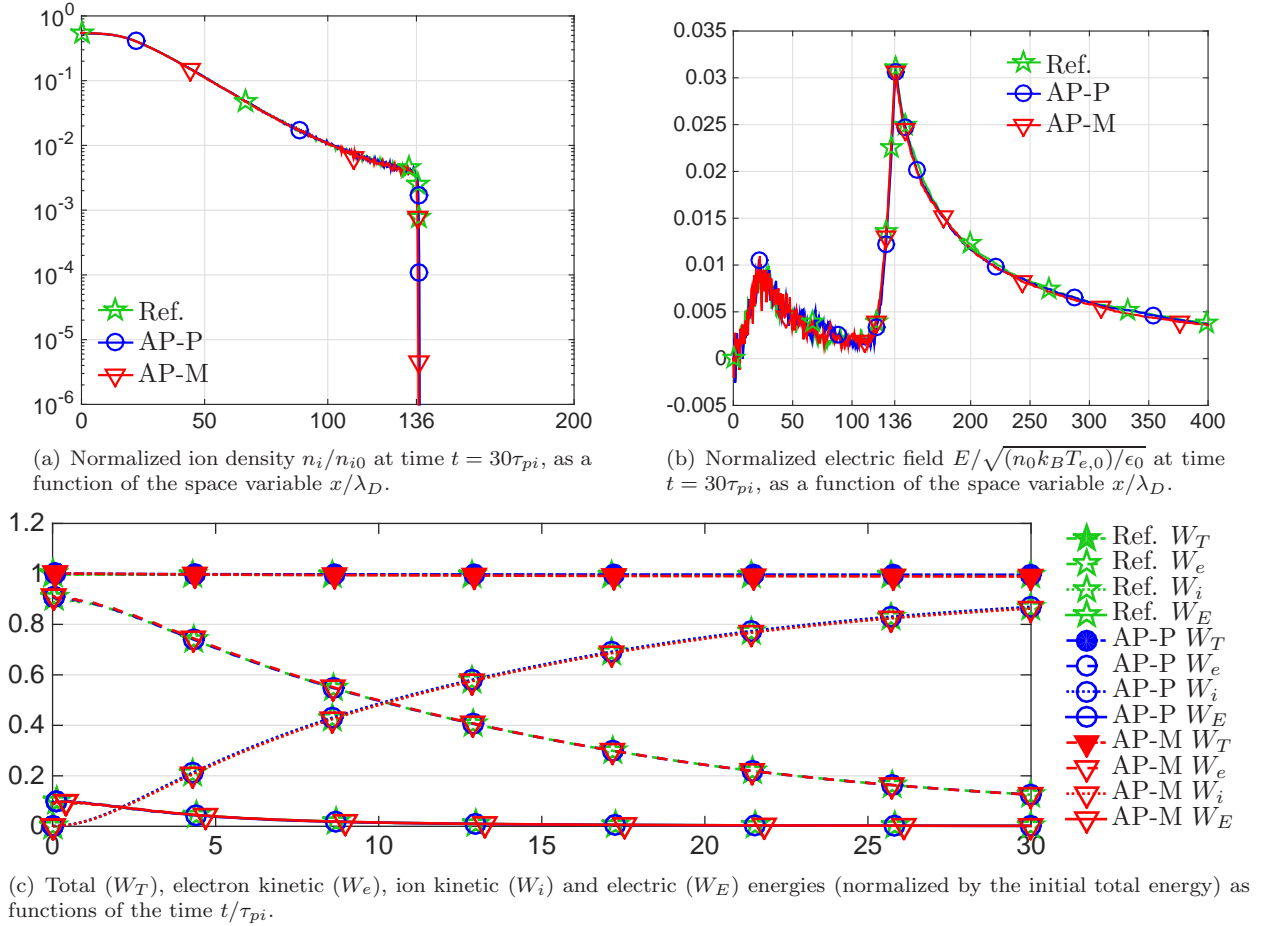

\begin{center}
\psfrag{Explicit}[l][l][1]{Ref.}
\psfrag{Particles}[l][l][1]{AP-P}
\psfrag{Moments}[l][l][1]{AP-M}
\subfigure[Normalized ion density $n_i/n_{i0}$ at time $t=30\tau_{pi}$, as a function of the space variable $x/\lambda_D$.]{%
\includegraphics[width=0.45\textwidth]{Expansion-Ref/RhoI.eps}}\hspace{0.06\textwidth}%
\subfigure[Normalized electric field $E/\sqrt{(n_0k_BT_{e,0})/\epsilon_0}$ at time $t=30\tau_{pi}$, as a function of the space variable $x/\lambda_D$.]{%
\includegraphics[width=0.45\textwidth]{Expansion-Ref/E.eps}}
\subfigure[Total ($W_T$), electron kinetic ($W_e$), ion kinetic ($W_i$) and electric ($W_E$) energies (normalized by the initial total energy) as functions of the time $t/\tau_{pi}$.\label{Expansion-ion-low-density:c}]{%
\psfrag{Ref-Wtotal}[l][l][1]{Ref. $W_T$}
\psfrag{Ref-Welec}[l][l][1]{Ref. $W_e$}
\psfrag{Ref-Wion}[l][l][1]{Ref. $W_i$}
\psfrag{Ref-Wfield}[l][l][1]{Ref. $W_E$}
\psfrag{Mom-Wtotal}[l][l][1]{AP-M $W_T$}
\psfrag{Mom-Welec}[l][l][1]{AP-M $W_e$}
\psfrag{Mom-Wion}[l][l][1]{AP-M $W_i$}
\psfrag{Mom-Wfield}[l][l][1]{AP-M $W_E$}
\psfrag{Part-Wtotal}[l][l][1]{AP-P $W_T$}
\psfrag{Part-Welec}[l][l][1]{AP-P $W_e$}
\psfrag{Part-Wion}[l][l][1]{AP-P $W_i$}
\psfrag{Part-Wfield}[l][l][1]{AP-P $W_E$}
\includegraphics[width=\textwidth]{Expansion-Ref/Energy-Ref.eps}}
\caption{Low-density plasma expansion. Results computed using the reference scheme (Ref.), the AP-Particle scheme (AP-P) and the AP-Moment scheme (AP-M). The simulations are performed with $\Delta x = 0.4 \lambda_D$, $\Delta t = 0.5 \tau_p$ and $10^6$ particles for each species.}
\label{Expansion-ion-low-density}
\end{center}
\end{figure}

A first series of simulations is performed with a low-density plasma (the initial density $n_{i0}$ is adjusted to obtain $\lambda_D=10^{-3}$ m and $\tau_p=10^{-3}$ s). The mesh size and the time step are smaller than the Debye length and the plasma period. The numerical results are represented in Figure~\ref{Expansion-ion-low-density} and, to facilitate the comparisons, they are expressed with the same normalized quantities as in \cite{GMAH08,DDNSV10}. The AP schemes reproduce correctly the physics of the plasma expansion: the plasma reaches the position $x_T\approx 136\lambda_D$ at the end of the simulation, there is a large electric field at the plasma edge, the (thermal) electron kinetic energy is converted into (drift) ion kinetic during the simulation. The total energy of the system is conserved, which demonstrates the good conservation properties of the schemes. Furthermore, the results of the AP schemes are almost identical to the results of the reference explicit scheme (the first one described in Section \ref{sec-exp-schemes}) and are in agreement with the results of \cite{GMAH08,DDNSV10}.

A second series of simulations is performed with a high-density plasma (the initial density $n_{i0}$ is adjusted to get $\lambda_D=10^{-4}$ m and $\tau_p=10^{-4}$ s). The reference explicit scheme is still used with discretization parameters resolving the Debye lenth and the plasma period, whereas the AP schemes are now used with discretization parameters significantly larger than the Debye length and the plasma period (the number of particles is also reduced). The Maxwell-Boltzmann relation used in the previous simulations to compute the initial electron density turns out to be inaccurate with a coarse mesh. It is thus better to make the initial electron density equal to the ion density. The numerical results are represented in Figure~\ref{Expansion-high-density}. As in the previous simulations, the AP schemes reproduce correctly the physics of the plasma expansion and the energy conservation is satisfactory. The results provided by the AP schemes are even quite close to the results obtained with the reference explicit scheme, despite the important difference of computational cost. Note that, in Figure~\ref{Expansion:high:rho}, the oscillations in the reference scheme results are not physical and are due to the insufficient number of particles. The outputs of the AP-Moment scheme and the AP-Particle schemes are similar but the discretization used for the AP-Moment scheme is finer. This shows that the AP-Particle scheme is less diffusive than the AP-Moment scheme and provides a better tracking of the plasma-vacuum interface. Finally, these simulations demonstrate the ability of the AP schemes to simulate quasi-neutral problems with moderate computational costs. However, it is not possible to use too coarse discretizations, owing to the deterioration of the energy conservation. This is a well-known drawback of the semi-implicit methods \cite{DGAH10}, due to fast particles crossing more than one cell in a time step. This point should be investigated in subsequent realizations.

\begin{figure}[!ht]
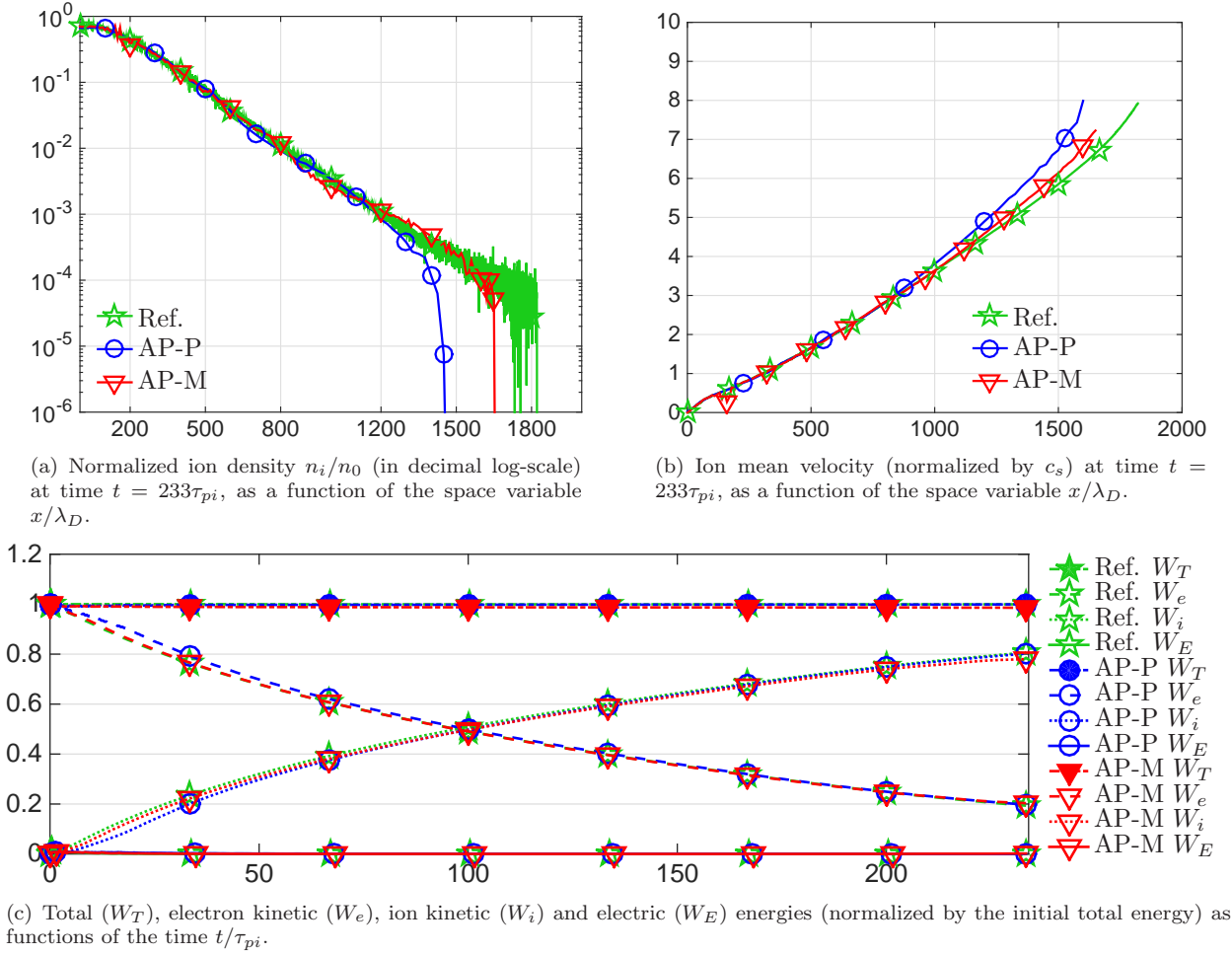

\begin{center}
\psfrag{Explicit}[l][l][1]{Ref.}
\psfrag{Particles}[l][l][1]{AP-P}
\psfrag{Moments}[l][l][1]{AP-M}
\subfigure[Normalized ion density $n_i/n_0$ (in decimal log-scale) at time $t=233 \tau_{pi}$, as a function of the space variable $x/\lambda_D$.\label{Expansion:high:rho}]{%
\includegraphics[width=0.45\textwidth]{Expansion-High/RhoI.eps}}\hspace{0.06\textwidth}%
\subfigure[Ion mean velocity (normalized by $c_s$) at time $t=233 \tau_{pi}$, as a function of the space variable $x/\lambda_D$.]{%
\includegraphics[width=0.45\textwidth]{Expansion-High/Ui.eps}\label{Expansion:high:Ui}}
\subfigure[Total ($W_T$), electron kinetic ($W_e$), ion kinetic ($W_i$) and electric ($W_E$) energies (normalized by the initial total energy) as functions of the time $t/\tau_{pi}$.]{%
\psfrag{Ref-Wtotal}[l][l][1]{Ref. $W_T$}
\psfrag{Ref-Welec}[l][l][1]{Ref. $W_e$}
\psfrag{Ref-Wion}[l][l][1]{Ref. $W_i$}
\psfrag{Ref-Wfield}[l][l][1]{Ref. $W_E$}
\psfrag{Mom-Wtotal}[l][l][1]{AP-M $W_T$}
\psfrag{Mom-Welec}[l][l][1]{AP-M $W_e$}
\psfrag{Mom-Wion}[l][l][1]{AP-M $W_i$}
\psfrag{Mom-Wfield}[l][l][1]{AP-M $W_E$}
\psfrag{Part-Wtotal}[l][l][1]{AP-P $W_T$}
\psfrag{Part-Welec}[l][l][1]{AP-P $W_e$}
\psfrag{Part-Wion}[l][l][1]{AP-P $W_i$}
\psfrag{Part-Wfield}[l][l][1]{AP-P $W_E$}
\includegraphics[width=\textwidth]{Expansion-High/Energy-P-25-5-6-M-10-2-6.eps}}
\caption{High-density plasma expansion. Results computed using the reference scheme (Ref.), the AP-Particle scheme (AP-P) and the AP-Moment scheme (AP-M). The discretization parameters are : $\Delta x = 0.5 \lambda_D$, $\Delta t = 0.2 \tau_p$ and $10^7$ particles (for each species) for the reference scheme ; $\Delta x = 25 \lambda_D$, $\Delta t = 5 \tau_p$ and $10^6$ particles for the AP-Particle scheme ; $\Delta x = 10 \lambda_D$, $\Delta t = 2 \tau_p$ and $10^6$ particles for the AP-Moment scheme.}
\label{Expansion-high-density}
\end{center}
\end{figure}




\subsection{A one-dimensional model of POS}

A Plasma Opening Switch (POS) is a device used to deliver a large current with a rapid increase of its impedance \cite{Del03}. It consists of a coaxial cylindrical transmission line filled whith a high-density plasma and connected to an input power generator \cite{zeng_particle--cell_2001}. In a first phase, the conduction phase, the plasma short-circuits the two electrodes of the transmission line and prevents the power to be delivered to the load. Then, the interaction of the electromagnetic wave with the plasma leads to the formation of a vacuum gap and finally to the total opening of the plasma, making possible the transmission of the current to the load. During the opening of the plasma, charge separation phenomena are essential and non-neutral sheath appear at the plasma edge \cite{fruchtman_sheath_1999, richardson_particle--cell_2012}.

A very simplified model of POS, one-dimensional in space and two-dimensional in velocity, is considered in this section. This model has been introduced in \cite{DDSprep}. Let us denote by $(x,y,z)$ the three-dimensional Cartesian coordinates. The particles move only along the $x$-direction but have a velocity both in the $x$-direction and the $y$-direction so that the electromagnetic field generated by the particles has only components $E_x$, $E_y$ and $B_z$. Therefore, the two-species Vlasov-Maxwell model reduces to the following equations:
\begin{align*}
&\partial_t f_e + v_x \cdot \partial_x f_e - \frac{e}{m_e}(E_x+ v_y B_z) \partial_{v_x} f_e -\frac{e}{m_e}(E_y- v_x B_z)\partial_{v_y} f_e =0\,,\\
&\partial_t f_i + v_x \cdot \partial_x f_i + \frac{e}{m_i}(E_x+ v_y B_z) \partial_{v_x} f_i +\frac{e}{m_i}(E_y- v_x B_z)\partial_{v_y} f_i =0\,,\\
& \frac{1}{c^2} \partial_t E_x  =- \mu_0 J_x\,,\quad \frac{1}{c^2} \partial_t E_y + \partial_x B_z  =-\mu_0 J_y\,,\quad \partial_t B_z + \partial_x E_y  = 0\,,
\end{align*}
supplemented with the Maxwell-Gauss and Maxwell-Thomson equations. In the above equations, the indices $i$ and $e$ denote the quantities related to the ions and the electrons, respectively.

The space domain is $[0,L]$ with $L=0.2$ m. At the initial time, the plasma fills the region between $x=0.05$ and $x=0.15$ with the density $n_0$ (different values of $n_0$ are used in the simulations). The initial ion and electron velocities follow Maxwellian distributions with the same temperature $T_i=T_e=10^3$ eV. The ion mass is such that $m_i= 2 \cdot 10^4$ $m_e$. A transverse electromagnetic wave is sent from the left end of the domain. The value of its electric component at $x=0$ and time $t$ is
\begin{equation}
E_{y,{\rm inc}}(t)=A_{\rm inc}\left(t/t_{\rm inc}\right)^{4/3}\left(8+\left(t/t_{\rm inc}\right)^4\left(\left(t/t_{\rm inc}\right)^8-6\right)\right)\,,
\end{equation}
where $A_{\rm inc}=1.8\cdot 10^8$ V$\cdot$m$^{-1}$ and $t_{\rm inc}=10^{-8}$ s. Transparent boundary conditions are prescribed at each end of the domain to avoid wave reflections.

\begin{table}[!ht]
\begin{center} 
\renewcommand{\arraystretch}{1.3}
\caption{One-dimensional model of POS. Plasma parameters and discretization parameters used in the different simulations. The initial density $n_0$ is in m$^{-3}$, the initial Debye length $\lambda_D$ in m and the plasma period $\tau_p$ in s. The value N$_\text{p}$ is the total number of particles.}
\label{tab-POS1D}
\begin{tabular}{| c || c |c |c || c| c|c|c|c|c| }\hline
\multirow{2}{*}{Config.} & \multirow{2}{*}{$n_0$} & \multirow{2}{*}{$\lambda_D$} & \multirow{2}{*}{$\tau_{p}$} & \multicolumn{3}{c|}{Reference scheme} & \multicolumn{3}{c|}{AP schemes} \\\cline{5-10}
 & & & & $\Delta x/ \lambda_D$ & $\Delta t/\tau_p$ & N$_\text{p}$ & $\Delta x / \lambda_D$ & $\Delta t/\tau_p$ & N$_\text{p}$ \\ \hline \hline
Low-a & $10^{16}$& $10^{-4}$ &$ 10^{-10}$ &$ 1 $ & $0.1$  &$10^6$&$1$ & $0.1$&$10^6$ \\\hline 
Low-b & \multirow{1}{*}{$10^{17}$}& \multirow{1}{*}{$5\; 10^{-5}$} &\multirow{1}{*}{$5 \;10^{-11}$} &\multirow{1}{*}{$0.2$} & \multirow{1}{*}{$0.04$}  &\multirow{1}{*}{$10^7$}&$20$ & $0.04$&$10^4$ \\\hline\hline
High-a & $10^{20}$& $ 10^{-5}$ &$10^{-11}$ & \multicolumn{3}{c|}{}&$10^2$ & $1$&$10^5$ \\ \cline{1-4} \cline{8-10}
High-b & $10^{22}$& $ 10^{-7}$ &$10^{-13}$ & \multicolumn{3}{c|}{}&$10^4$ & $10^2$ & $10^5$ \\\hline
\end{tabular}
 \end{center}
\end{table}

Four series of simulations are performed (their main characteristics are collected in Table~\ref{tab-POS1D}). In the Low-a simulations, the AP schemes and the reference explicit scheme (the second one described in Section \ref{sec-exp-schemes}) are used with discretization parameters that resolve the Debye length and the plasma period. The numerical results, represented in Figure~\ref{fig:POS:low}, are comparable for the three schemes. In particular, the level of numerical noise is equivalent. We observe that the electrons and the ions are accelerated by the incident wave (see Figures~\ref{fig:POS:low}(c) and \ref{fig:POS:low}(d)). The electrons being less massive than the ions, they are accelerated more strongly and are expelled from the plasma edge, which breaks the quasi-neutrality in this zone (see Figures~\ref{fig:POS:low}(a) and \ref{fig:POS:low}(b)). This charge separation creates a large electric field $E_x$ (see Figure~\ref{fig:POS:low}(e)). As for the magnetic field, it is significantly transmitted through the plasma (see Figure~\ref{fig:POS:low}(f)).

\begin{figure}
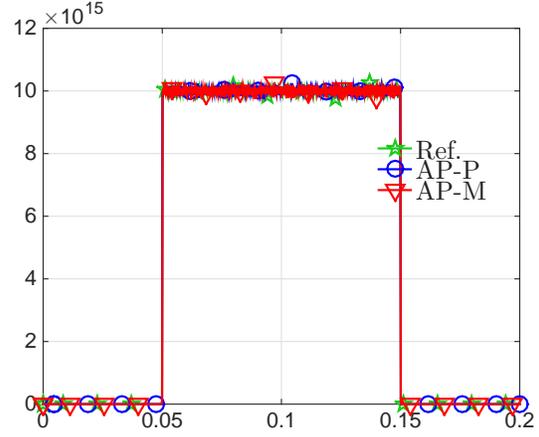
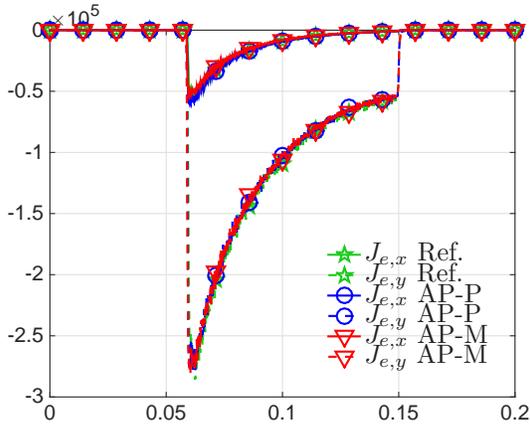
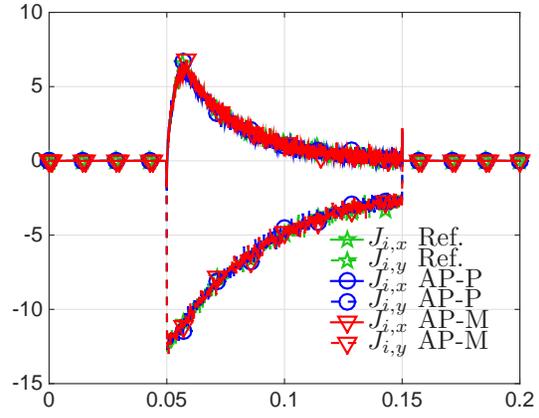
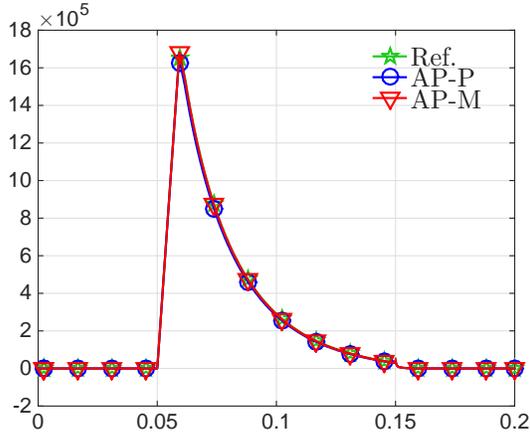
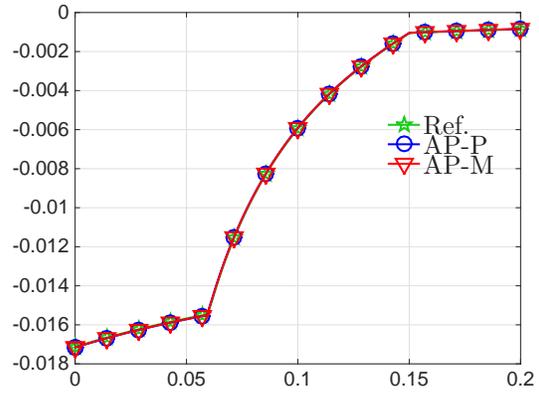

\begin{center}
\subfigure[Electron density (in m$^{-3}$).]{
\psfrag{(Ref)}[l][l][1.]{Ref.}
\psfrag{(Part)}[l][l][1.]{AP-P}
\psfrag{(Moment)}[l][l][1.]{AP-M}
\includegraphics[width=0.42\textwidth]{./P1-low-a/RhoE.eps}}\hspace*{0.06\textwidth}%
\subfigure[Ion density (in m$^{-3}$).]{
\psfrag{(Ref)}[l][l][1.]{Ref.}
\psfrag{(Part)}[l][l][1.]{AP-P}
\psfrag{(Moment)}[l][l][1.]{AP-M}
\includegraphics[width=0.42\textwidth]{./P1-low-a/RhoI.eps}}

\subfigure[Electron current density in the $x$-direction and $y$-direction (in A$\cdot \text{m}^{-2}$).]{
\psfrag{(X Ref)}[l][l][1.]{$J_{e,x}$ Ref.}
\psfrag{(Y Ref)}[l][l][1.]{$J_{e,y}$ Ref.}
\psfrag{(X Part)}[l][l][1.]{$J_{e,x}$ AP-P}
\psfrag{(Y Part)}[l][l][1.]{$J_{e,y}$ AP-P}
\psfrag{(X Moment)}[l][l][1.]{$J_{e,x}$ AP-M}
\psfrag{(Y Moment)}[l][l][1.]{$J_{e,y}$ AP-M}
\includegraphics[width=0.42\textwidth]{./P1-low-a/Jexy.eps}}\hspace*{0.06\textwidth}%
\subfigure[Ion current density in the $x$-direction and $y$-direction (in A$\cdot \text{m}^{-2}$).]{
\psfrag{(X Ref)}[l][l][1.]{$J_{i,x}$ Ref.}
\psfrag{(Y Ref)}[l][l][1.]{$J_{i,y}$ Ref.}
\psfrag{(X Part)}[l][l][1.]{$J_{i,x}$ AP-P}
\psfrag{(Y Part)}[l][l][1.]{$J_{i,y}$ AP-P}
\psfrag{(X Moment)}[l][l][1.]{$J_{i,x}$ AP-M}
\psfrag{(Y Moment)}[l][l][1.]{$J_{i,y}$ AP-M}
\includegraphics[width=0.42\textwidth]{./P1-low-a/Jixy.eps}}

\subfigure[Electric field $E_x$ (V$\cdot\text{m}^{-1}$).]{
\psfrag{(Ref)}[l][l][1.]{Ref.}
\psfrag{(Part)}[l][l][1.]{AP-P}
\psfrag{(Moment)}[l][l][1.]{AP-M}
\includegraphics[width=0.42\textwidth]{./P1-low-a/Ex.eps}}\hspace*{0.06\textwidth}%
\subfigure[Magnetic field $B_z$ (in Teslas).]{
\psfrag{(Ref)}[l][l][1.]{Ref.}
\psfrag{(Part)}[l][l][1.]{AP-P}
\psfrag{(Moment)}[l][l][1.]{AP-M}
\includegraphics[width=0.42\textwidth]{./P1-low-a/Bz.eps}}

\caption{Low-density POS (Low-a). Results computed using the reference scheme (Ref.), the AP-Moment scheme (AP-M) and the AP-Particle scheme (AP-P) at time $t= 3\cdot 10^{-9}$ s.}\label{fig:POS:low}
\end{center}
\end{figure}

The Low-b simulations deal with a higher initial plasma density. The reference explicit scheme still uses discretization parameters that resolve the Debye length and the plasma period, while the AP schemes use a mesh size 20 times larger than the Debye length and a reduced number of particles. Despite the huge difference of computational cost between the simulations, the numerical results are indistinguishable (see Figure~\ref{fig:POS:low:b}). In these simulations, the current created by the electron motion at the plasma edge is strong enough to stop the penetration of the magnetic field into the plasma bulk (see Figure~\ref{fig:POS:low:b}(d)).

\begin{figure}
\begin{center}
\subfigure[Electron density (in m$^{-3}$).]{
\psfrag{(Ref)}[l][l][1.]{Ref.}
\psfrag{(Part)}[l][l][1.]{AP-P}
\psfrag{(Moment)}[l][l][1.]{AP-M}
\includegraphics[width=0.42\textwidth]{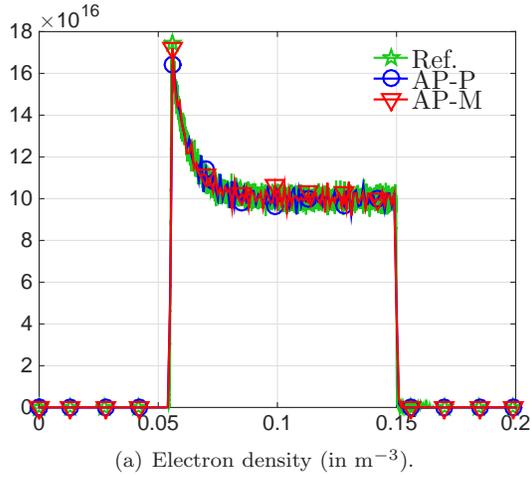}}\hspace*{0.06\textwidth}%
\subfigure[Electron current density in the $x$-direction and $y$-direction (in A$\cdot \text{m}^{-2}$).]{
\psfrag{(X Ref)}[l][l][1.]{$J_{e,x}$ Ref.}
\psfrag{(Y Ref)}[l][l][1.]{$J_{e,y}$ Ref.}
\psfrag{(X Part)}[l][l][1.]{$J_{e,x}$ AP-P}
\psfrag{(Y Part)}[l][l][1.]{$J_{e,y}$ AP-P}
\psfrag{(X Moment)}[l][l][1.]{$J_{e,x}$ AP-M}
\psfrag{(Y Moment)}[l][l][1.]{$J_{e,y}$ AP-M}
\includegraphics[width=0.42\textwidth]{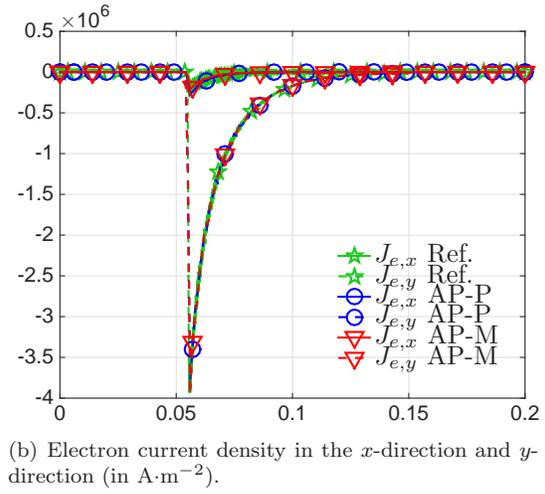}}

\subfigure[Electric field $E_x$ (in V$\cdot\text{m}^{-1}$).]{
\psfrag{(Ref)}[l][l][1.]{Ref.}
\psfrag{(Part)}[l][l][1.]{AP-P}
\psfrag{(Moment)}[l][l][1.]{AP-M}
\includegraphics[width=0.42\textwidth]{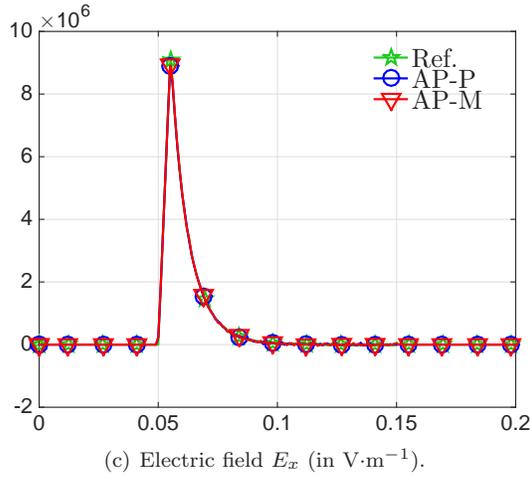}}\hspace*{0.06\textwidth}%
\subfigure[Magnetic field $B_z$ (in Teslas).]{
\psfrag{(Ref)}[l][l][1.]{Ref.}
\psfrag{(Part)}[l][l][1.]{AP-P}
\psfrag{(Moment)}[l][l][1.]{AP-M}
\includegraphics[width=0.42\textwidth]{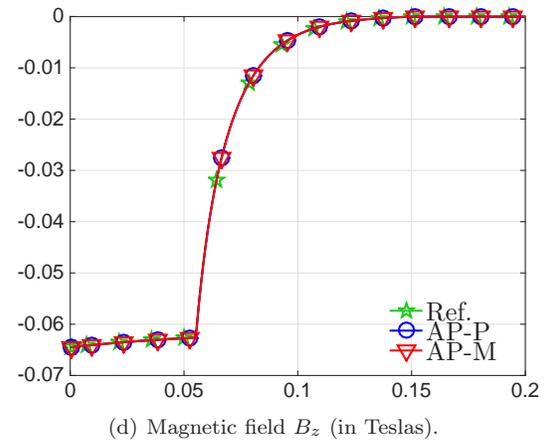}}

\caption{Low-density POS (Low-b). Results computed using the reference scheme (Ref.), the AP-Moment scheme (AP-M) and the AP-Particle scheme (AP-P) at time $t= 4 \cdot 10^{-9} $ s.}\label{fig:POS:low:b}
\end{center}
\end{figure}

For the high-density cases (High-a and High-b), the computational cost of the reference explicit scheme is prohibitive, so that no simulation is carried out with this scheme. The AP schemes are used with very coarse discretizations. For the most demanding case, the mesh size is $10^4$ larger than the Debye length, the time step is $100$ times larger than the plasma period and the number of numerical particles is only $10^5$. The numerical results are shown in Figure~\ref{fig:POS:high}. In contrast to the low-density simulations, the incident magnetic field is almost entirely reflected at the plasma edge (due to the very large electron current). The level of numerical noise in the outputs remains moderate, even in the High-b simulation. These simulations demonstrate the ability of the AP schemes to handle high-density plasmas and vacuum-plasma transitions with very coarse discretizations.

\begin{figure}
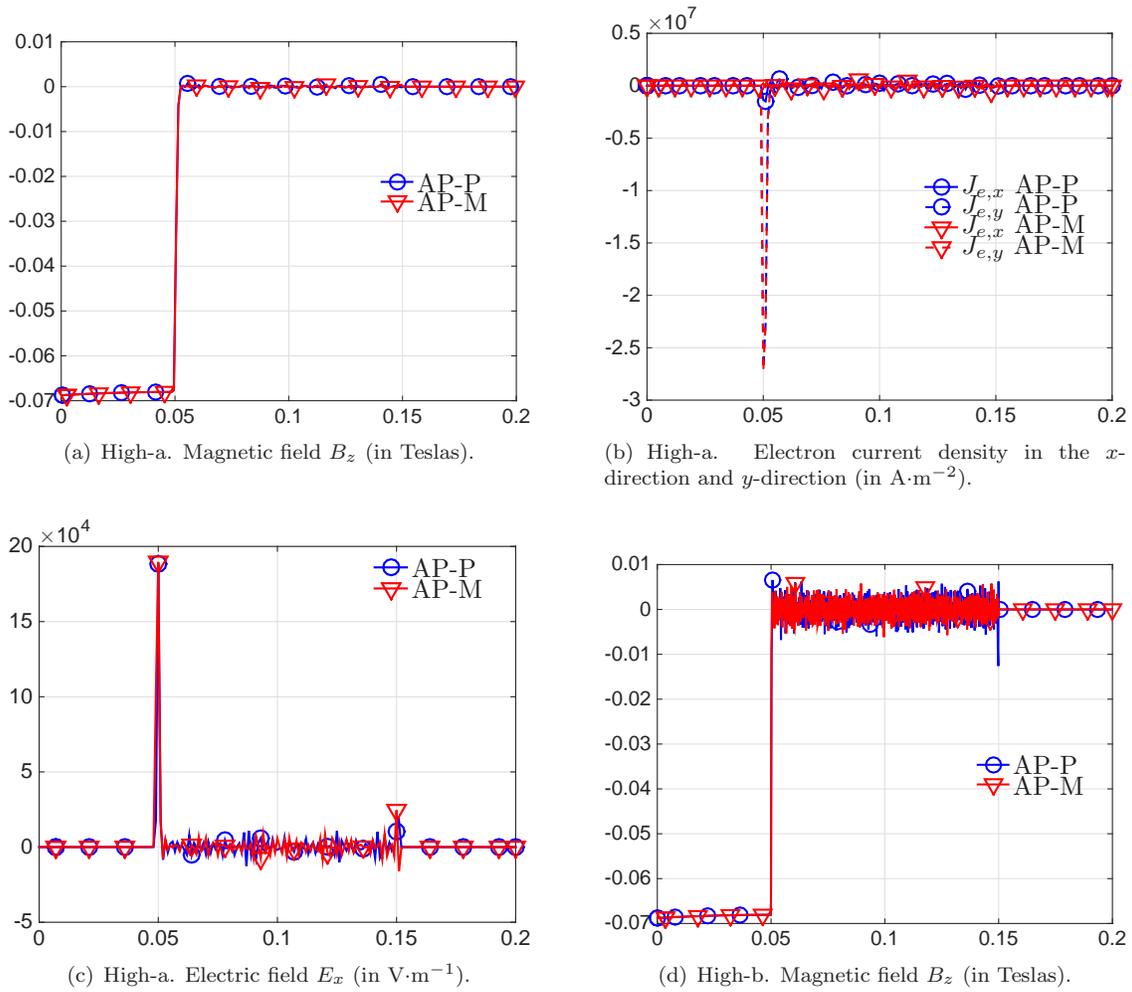

\begin{center}
\subfigure[High-a. Magnetic field $B_z$ (in Teslas).]{%
\psfrag{(Ref)}[l][l][1.]{Exp.}
\psfrag{(Part)}[l][l][1.]{AP-P}
\psfrag{(Moment)}[l][l][1.]{AP-M}
\includegraphics[width=0.42\textwidth]{./P1-high-a/Bz.eps}}\hspace*{0.06\textwidth}%
\subfigure[High-a. Electron current density in the $x$-direction and $y$-direction (in A$\cdot \text{m}^{-2}$).]{%
\psfrag{(X Ref)}[l][l][1.]{$J_{e,x}$ Ref.}
\psfrag{(Y Ref)}[l][l][1.]{$J_{e,y}$ Ref.}
\psfrag{(X Part)}[l][l][1.]{$J_{e,x}$ AP-P}
\psfrag{(Y Part)}[l][l][1.]{$J_{e,y}$ AP-P}
\psfrag{(X Moment)}[l][l][1.]{$J_{e,x}$ AP-M}
\psfrag{(Y Moment)}[l][l][1.]{$J_{e,y}$ AP-M}
\includegraphics[width=0.42\textwidth]{./P1-high-a/Jexy.eps}}

\subfigure[High-a. Electric field $E_x$ (in V$\cdot\text{m}^{-1}$).]{%
\psfrag{(Ref)}[l][l][1.]{Exp.}
\psfrag{(Part)}[l][l][1.]{AP-P}
\psfrag{(Moment)}[l][l][1.]{AP-M}
\includegraphics[width=0.42\textwidth]{./P1-high-a/Ex.eps}}\hspace*{0.06\textwidth}%
\subfigure[High-b. Magnetic field $B_z$ (in Teslas).]{%
\psfrag{(Part)}[l][l][1.]{AP-P}
\psfrag{(Moment)}[l][l][1.]{AP-M}
\includegraphics[width=0.42\textwidth]{./P1-high-b/Bz.eps}}
\caption{High-density POS (High-a and High-b). Results computed using the AP-Moment scheme (AP-M) and the AP-Particle scheme (AP-P) at time $t= 4 \cdot 10^{-9} $ s.}\label{fig:POS:high}
\end{center}
\end{figure}

\subsection{Propagation of a KMC wave in a POS}

We now consider a two-dimensional model of POS where the plasma density is not uniform in the transverse direction to the electromagnetic wave propagation. In this configuration a magnetic shock wave, the so-called KMC wave, propagates into the plasma \cite{swanekamp_particlecell_1996}. The existence of KMC waves can be derived from the quasi-neutral equations (see \ref{sec:KMC2}). Consequently, their simulation offers a means to verify the consistency with the quasi-neutral limit in an electromagnetic context.

The model is a one-species, two-dimensional model. The ions are motionless but their density is not uniform. Let us denote by $(x,y,z)$ the three-dimensional Cartesian coordinates. The particle motion is restricted to the $(x,y)$-plane and thus the electromagnetic field generated by the particles has only components $E_x$, $E_y$, and $B_z$. Finally, the Vlasov-Maxwell system simplifies into
\begin{align}
&\partial_t f + v_x \cdot \partial_x f  + v_y \cdot \partial_y f - \frac{e}{m}(E_x+ v_y B_z) \partial_{v_x} f -\frac{e}{m}(E_y- v_x B_z)\partial_{v_y} f =0\,,\label{2DVM1}\\
& \frac{1}{c^2} \partial_t E_x - \partial_y B_z =- \mu_0 J_x\,,\quad \frac{1}{c^2} \partial_t E_y + \partial_x B_z  =-\mu_0 J_y\,,\quad \partial_t B_z + \partial_x E_y - \partial_y E_x = 0\,,\label{2DVM2}
\end{align}
supplemented with the Maxwell-Gauss and Maxwell-Thomson equations.

The space domain is the rectangle $[0,L_x]\times[0,L_y]$ with $L_x=0.2$ m and $L_y=0.03$ m. The lower and upper sides of the domain represent the cathode and the anode of the POS, respectively. At the initial time, the plasma fills the area between $x=0.05$ and $x=0.15$ with the density
\begin{equation*} 
n_0(x,y)= \begin{cases}n_{\rm min} & \text{if $0\leq y < 0.01$} \\
n_{\rm min}n_{\rm max}/[(n_{\rm min}-n_{\rm max})(y/0.01-1)+ n_{\rm max}] & \text{if $0.01 \leq y < 0.02$} \\
 n_{\rm min} & \text{if $0.02\leq y \leq 0.03$}
\end{cases}
\end{equation*}
A transverse electromagnetic wave is sent from the left end of the domain. Transparent boundary conditions are prescribed at each end of the domain to avoid wave reflections.

The initial plasma density is such that $\partial_y \left(\frac{1}{n_0}\right) (x,y)=10^2 (n_{\rm max}-n_{\rm min})/(n_{\rm max}n_{\rm min})$ for $0.01 \leq y \leq 0.02$. Therefore, according to the theory presented in \ref{sec:KMC2}, the incident wave should propagate into the plasma as a KMC wave (between the lines $y=0.01$ and $y=0.02$). Moreover, the expected speed of this KMC wave is
\begin{equation}\label{eq:def:KMC:Speed2}
  V_\text{KMC} = 10^2 \frac{B_z}{e \mu_0} \frac{n_{\rm max}-n_{\rm min}}{n_{\rm max}n_{\rm min}}\,.
\end{equation}

\begin{table}
 \centering
 \renewcommand{\arraystretch}{1.3}
  \caption{Propagation of a KMC wave in a POS. Plasma parameters and discretization parameters used in the different simulations. The inital densities $n_{\rm min}$ and $n_{\rm max}$ are in m$^{-3}$, the electron temperature $T_e$ in eV, the initial Debye length $\lambda_D$ in m, the theoretical KMC wave speed in m$\cdot$s$^{-1}$. The value $\text{N}_\text{p}$ is the number of particles. }
  \label{tab:KMC}
  \begin{tabular}[c]{|c||c|c|c|c|c|c|c|c|c|}\hline
Config. & $n_{\rm min}$ & $n_{\rm max}$  & $T_e$ & $\lambda_D$ &  $V_\text{KMC}$ & Grid & $\Delta x$ & $\Delta y$ & $\text{N}_{\text{p}}$ \\\hline \hline
   (a) &$10^{19}$ & $10^{20}$ & $6 \cdot 10^4$ &$1.69 \cdot 10^{-6}$ & $25 \cdot 10^6$  & $100 \times 100$ & $ 2 \cdot 10^{-3}$ & $3\cdot 10^{-4}$ & $10^5$ \\\hline
   (b) & $10^{19}$ & $10^{21}$ & $6 \cdot 10^2$ & $5.35 \cdot 10^{-8}$ & $25 \cdot 10^6$  &  $400 \times 100$ &$ 5 \cdot 10^{-4}$ & $3\cdot 10^{-4}$ & $4\cdot 10^{6}$\\\hline
  \end{tabular}
\end{table}

Two simulations are carried out on this problem, both with the AP-Moment scheme (note that the Gauss elliptic correction is not implemented in these simulations). The plasma and discretization parameters are specified in Table~\ref{tab:KMC}. The evolution of the magnetic field for the simulation (a), reported in Figure \ref{fig:KMC}, shows a rapid magnetization of the plasma in the region where the density gradient is located. The electrons emitted at the cathode produce a current which prevents a uniform penetration of the magnetic field in the plasma. This current is gradually deflected with the propagation of the magnetic field, as depicted in Figures \ref{fig:KMC:c} and \ref{fig:KMC:d}. We evaluate the speed of the KMC waves in the simulations by marking the position of a certain magnetic level set at different times (see Figure~\ref{fig2:KMC}). The observed values are in agreement with the theoretical values: in the simulation (a), it is 76\% of the theoretical value, while in the simulation (b), it reaches 82\% of the theoretical value. These good numerical results are obtained with discretizations that do not resolve the Debye length. In particular, for the discretization (a), the mesh size is $10^4$ times larger than the Debye length. 

\begin{figure}[!ht]
  \centering
  \subfigure[Magnetic field $B_z$ (in Teslas) at time $t_4=2.58$ ns. \label{fig:KMC:a}]{%
  \includegraphics[width=0.9\textwidth]{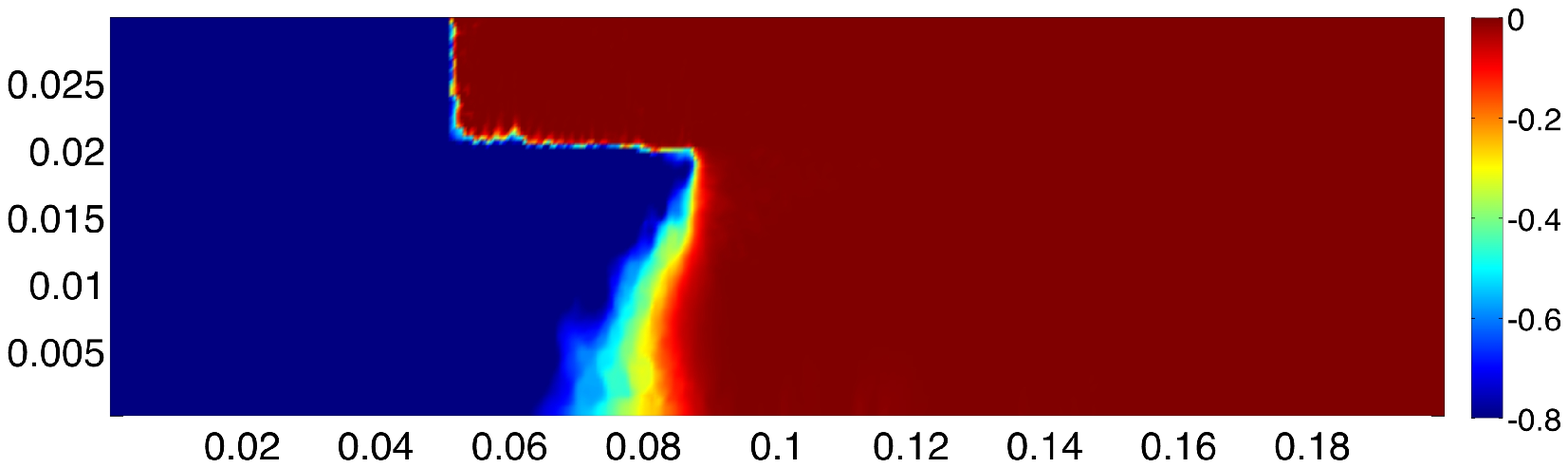}}
  \subfigure[Magnetic field $B_z$ (in Teslas) at time $t_6=3.67$ ns. \label{fig:KMC:b}]{%
  \includegraphics[width=0.9\textwidth]{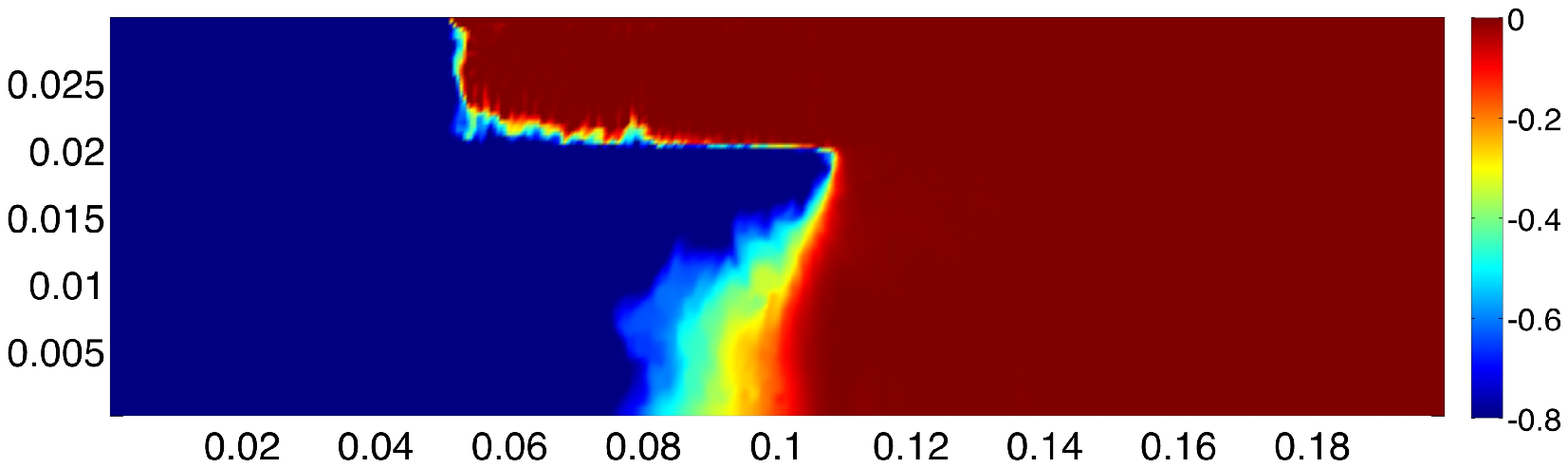}}
  \subfigure[Current density $J_e$ (in A$\cdot$ m$^{-2}$, in decimal log-scale) at time $t_4=2.58$ ns. \label{fig:KMC:c}]{%
  \includegraphics[width=0.4\textwidth]{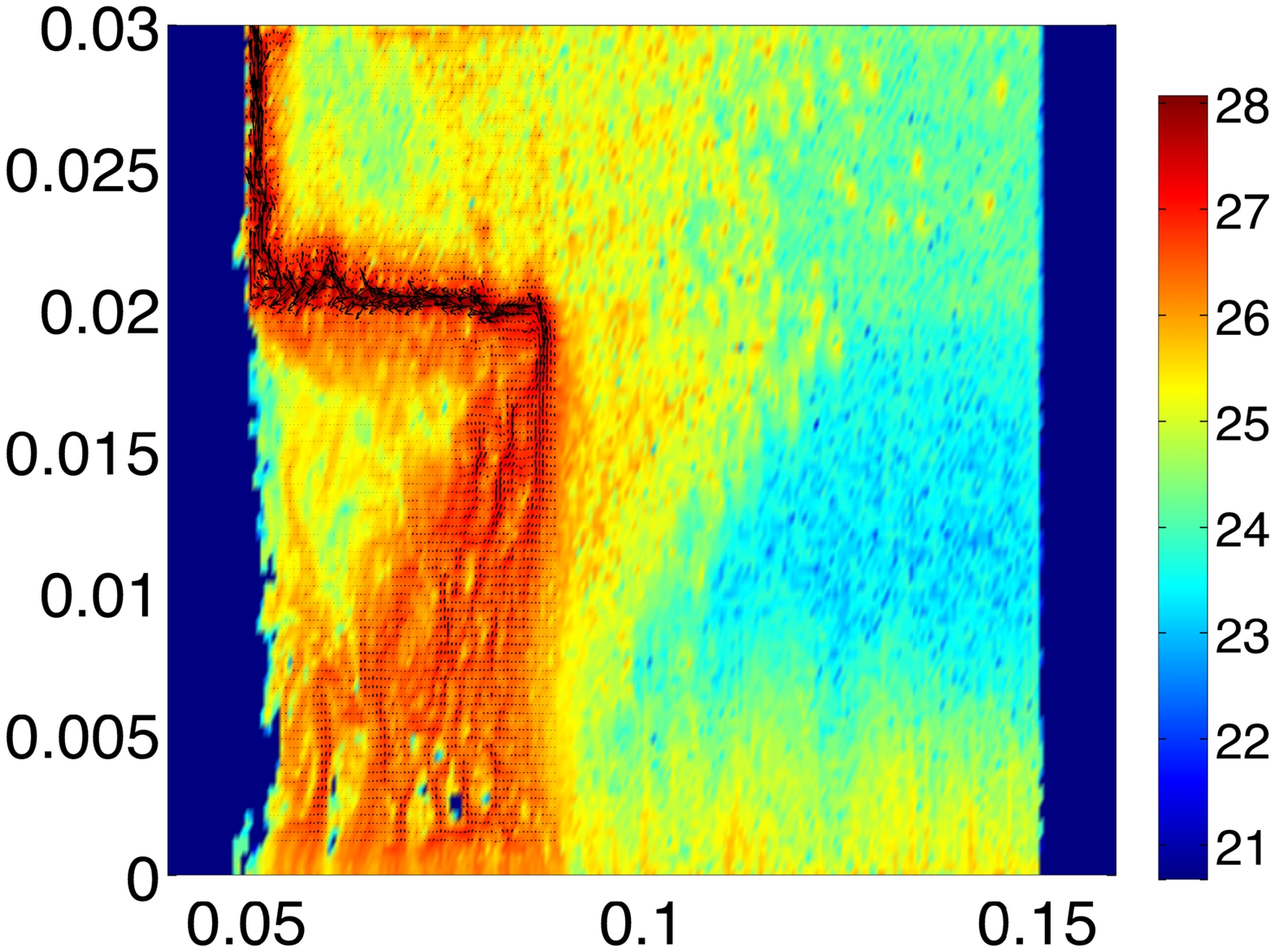}}\hspace*{0.05\textwidth}%
  \subfigure[Current density $J_e$ (in A$\cdot$ m$^{-2}$, in decimal log-scale) at time $t_6=3.67$ ns. \label{fig:KMC:d}]{%
    \includegraphics[width=0.4\textwidth]{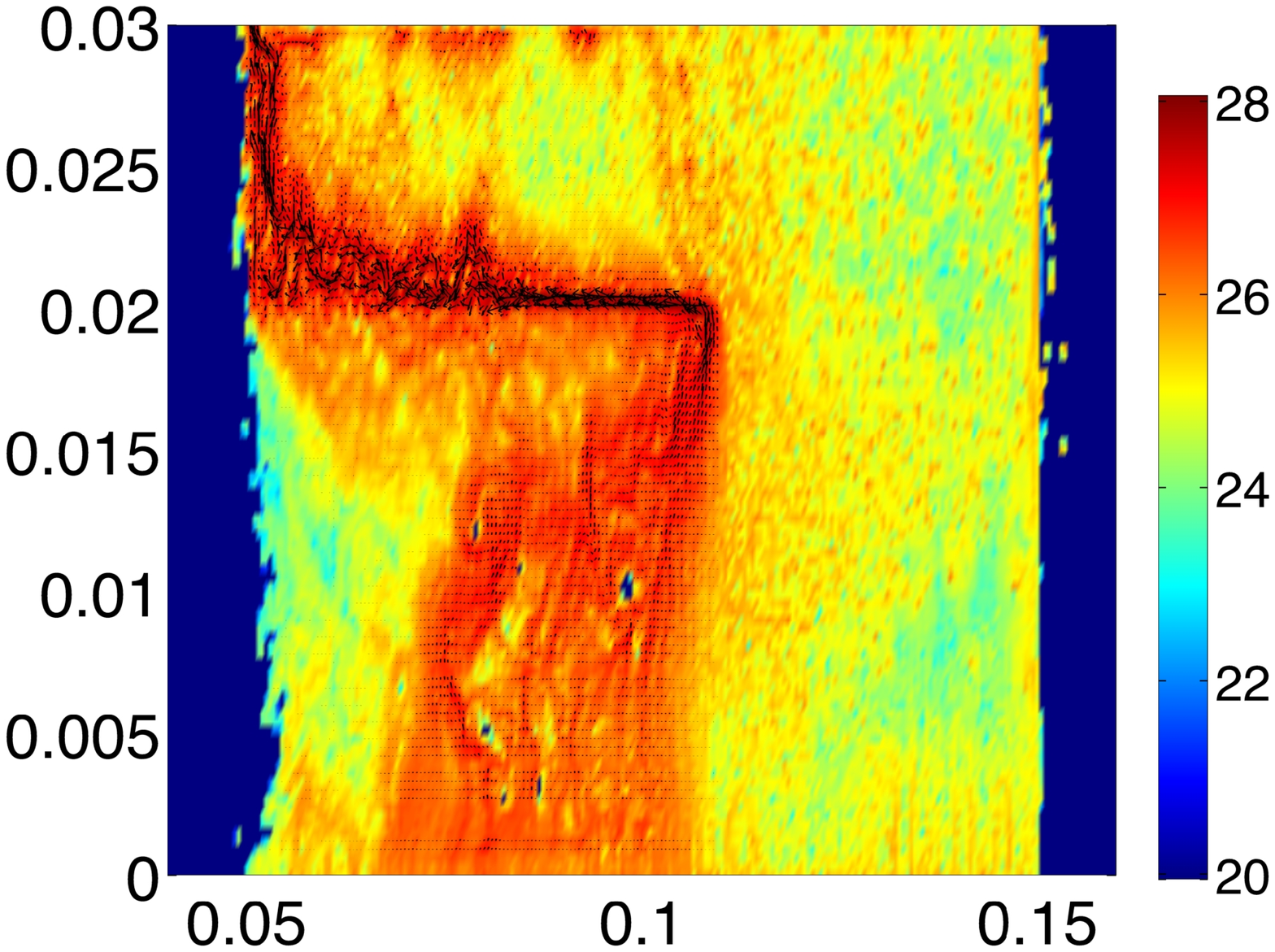}}
  \caption{Propagation of a KMC wave in a POS. Numerical results computed using the AP-Moment scheme, with discretization (b).}
  \label{fig:KMC}
\end{figure}

\begin{figure}[ht]
  \centering
  \subfigure[Configuration (a). Level sets of the magnetic field at times $t_1^{(a)}=1.75$, $t_2^{(a)}= 2.16$, $t_3^{(a)}=2.56$, $t_4^{(a)}=3.03$, $t_5^{(a)}=3.56$ and $t_6^{(a)}=4.09$ ns. \label{tab:KMC:a}]{%
    \psfrag{ T05}[l][l][1.]{ $t^{(a)}_1$}
    \psfrag{ T10}[l][l][1.]{ $t^{(a)}_2$}
    \psfrag{ T15}[l][l][1.]{ $t^{(a)}_3$}
    \psfrag{ T20}[l][l][1.]{ $t^{(a)}_4$}
    \psfrag{ T25}[l][l][1.]{ $t^{(a)}_5$}
    \psfrag{ T30}[l][l][1.]{ $t^{(a)}_6$}
    \includegraphics[width=0.8\textwidth]{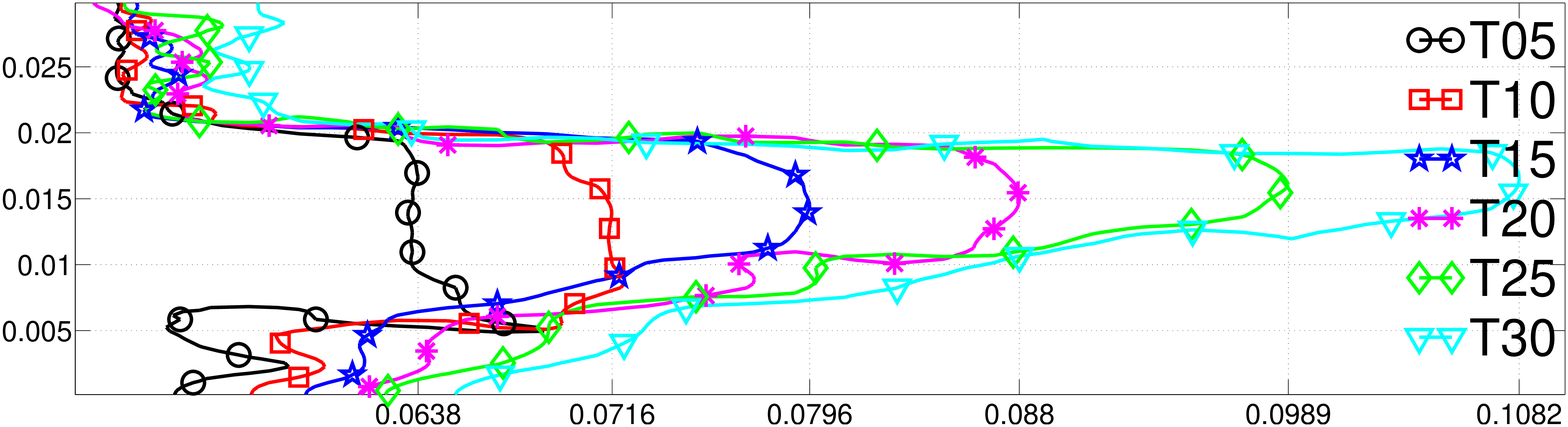}}

  \subfigure[Configuration (b). Level sets of the magnetic field at times $t_1^{(b)}=1.24$, $t_2^{(b)}= 1.62$, $t_3^{(b)}=2.06$, $t_4^{(b)}=2.58$, $t_5^{(b)}=3.11$ and $t_6^{(b)}=3.67$ ns.\label{tab:KMC:b}]{%
   \psfrag{ T05}[l][l][1.]{ $t^{(b)}_1$}
    \psfrag{ T10}[l][l][1.]{ $t^{(b)}_2$}
    \psfrag{ T15}[l][l][1.]{ $t^{(b)}_3$}
    \psfrag{ T20}[l][l][1.]{ $t^{(b)}_4$}
    \psfrag{ T25}[l][l][1.]{ $t^{(b)}_5$}
    \psfrag{ T30}[l][l][1.]{ $t^{(b)}_6$}
    \includegraphics[width=0.8\textwidth]{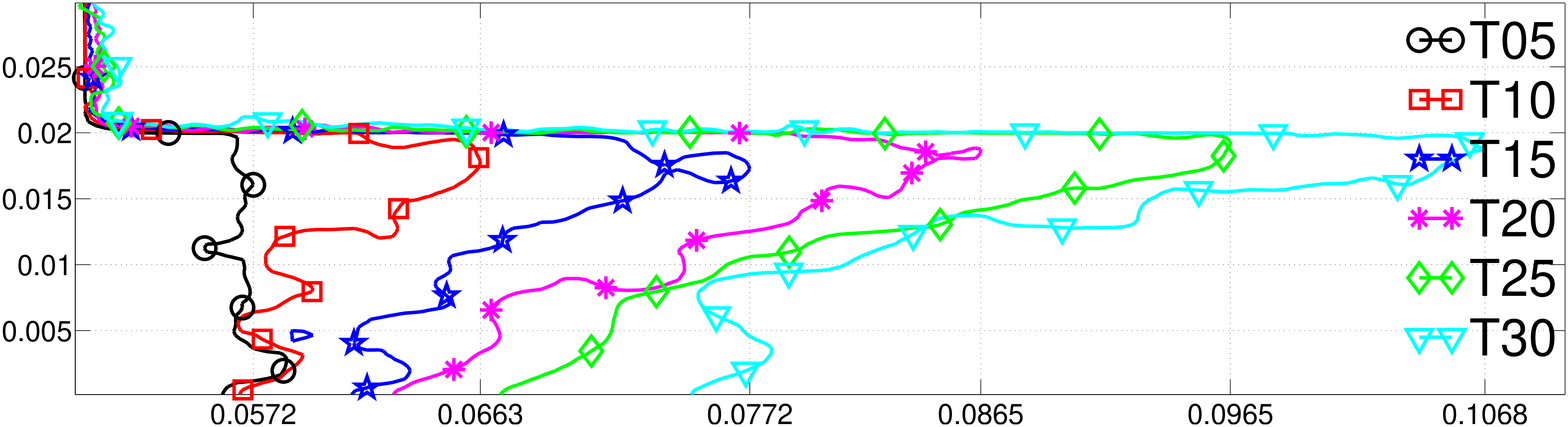}}

\renewcommand{\arraystretch}{1.2}
  \begin{tabular}[c]{|c|c|c|c|c|c|}\hline
  Config. & $V_{1-2}$ & $V_{2-3}$ & $V_{3-4}$ & $V_{4-5}$ & $V_{5-6}$ \\\hline
   (a) & $19.4\cdot 10^6$ & $19.8\cdot 10^6$ & $17.8\cdot 10^6$ & $20.7\cdot 10^6$ & $17.3\cdot 10^6$  \\\hline
   (b)& $24.3\cdot 10^6$ & $24.8\cdot 10^6$ & $17.9\cdot 10^6$ & $18.8\cdot 10^6$ & $18.2\cdot 10^6$  \\\hline
  \end{tabular}
  \caption{Propagation of a KMC wave in a POS. Estimation of the KMC wave speed in the simulations (a) and (b). The velocities $V_{i-j}$ are the estimated velocities on the time interval $[t_i, t_j]$. They are computed using the position of the level set $B_z=-0.8$ Tesla. }
  \label{fig2:KMC}
\end{figure}
\section{Conclusion}

We have derived two Asymptotic-Preserving Particle-In-Cell methods for the Vlasov-Maxwell system in the quasi-neutral limit. The scaling assumptions made for the definition of the quasi-neutral limit, similar to those used to derive the MHD models, yield a kinetic quasi-neutral model where the electric field is computed by means of a generalized Ohm law. The Asymptotic-Preserving methods are consistent with either the quasi-neutral model or the Vlasov-Maxwell model according to how the discretization parameters resolve the plasma parameters, which make them able to simulate complex plasma problems with a reasonable computational cost. No rigorous numerical analysis is provided in this article and this should be the subject of future work. However, the numerous numerical investigations demonstrate conclusively the efficiency of the methods to account for phenomena evolving at the plasma period and Debye length scales, as well as quasi-neutral phenomena usually well described by the MHD theory. In particular, they are able to cope with vacuum-dense plasma interfaces and the formation of non-neutral sheaths. Although other semi-implicit Particle-In-Cell methods (the Direct Implicit and Implicit Moment methods) present the same kind of Asymptotic-Preserving properties, the methodology developed in this article is new and provides a rigorous framework for the quasi-neutral limit problem. Furthermore, it opens the way for addressing more singular asymptotics and deriving more efficient numerical methods for some kinds of problems.

\medskip

\section*{Acknowledgments} This work has been carried out within the framework of the EUROfusion Consortium and has received funding from the Euratom research and training programme 2014-2018 under grant agreement No 633053. The views and opinions expressed herein do not necessarily reflect those of the European Commission.\\
Furthermore, the authors would like to acknowledge support from the ``Fondation Sciences et Technologies pour l'A\'eronautique et l'Espace'', in the frame of the project ``Plasmax'', the ANR PEPPSI  (Plasma Edge Physics and Plasma-Surface Interactions, 2013-2017) and ANR MOONRISE (MOdels, Oscillations and NumeRIcal SchEmes, 2015-2019).
This work has been supported by the National Science Foundation (NSF) under
grant RNMS11-07444 (KI-Net). PD is on leave from CNRS, Institut de
Math\'ematiques de Toulouse, France. PD acknowledges support from the Royal
Society and the Wolfson foundation through a Royal Society Wolfson Research
Merit Award.

\appendix

\section{Yee finite differences and $\theta$-scheme for the Maxwell equations}
\label{sec-prop-theta}

We consider the discretization of the homogeneous $d$-dimensional Maxwell equations ($d=$1, 2 or 3) with Yee finite differences for the space approximation and a $\theta$-scheme for the time integration. Using the discrete operators introduced in Section \ref{sec:defnum}, the discrete equations are
\begin{align}
& \frac{1}{c^2}\frac{E^{m+1}_h-E^{m}_h}{\Delta t} - \nabla_h \times \bar{B}^{m+\theta}_h =0\,,\label{theta1}\\
& \frac{B^{m+1}_h-B^{m}_h}{\Delta t} + \nabla_h \times \bar{E}^{m+\theta}_h=0\,,\label{theta2}
\end{align}
with
\begin{equation*}
\bar{E}^{m+\theta}_h = \theta E^{m+1}_h + (1-\theta) E^{m}_h \,, \qquad \bar{B}^{m+\theta}_h = \theta B^{m+1}_h + (1-\theta) {B}^{m}_h\,.
\end{equation*}
The above scheme is second-order accurate in space. It is first-order accurate in time for $\theta \neq \demi$ and second-order accurate for $\theta = \demi$. The following energy balance holds true:
\begin{equation}\label{theta-balance}
\E_h^{m+1} - \E_h^{m}= - \left(\theta-\demi \right)\left(\epsilon_0\left(E^{m+1}_h-E^{m}_h\right)^2+\frac{1}{\mu_0}\left(B^{m+1}_h-B^{m}_h\right)^2\right)\,,
\end{equation}
where $\E_h^{m}= \frac{\epsilon_0}{2} \|E^{m}_h\|^2 + \frac{1}{2\mu_0} \|B^{m}_h\|^2$. Therefore, the $\theta$-scheme is unconditionally stable for $\theta \in [\demi,1]$. It is dissipative for $\theta \in ]\demi,1]$ (all the more dissipative than $\theta$ is large) and energy-conserving for $\theta=\demi$. Unlike the leap-frog scheme, the $\theta$-scheme is dispersive even in a one-dimensional setting and with a Courant number equal to 1.

\section{Comparison of some particle pushers for the AP schemes}\label{sec:Part:Pusher}

We examine the properties of some particle pushers within the AP schemes. The particle pushers we consider are variants of the Boris scheme:
\begin{align}
&\frac{X_{N,j}^{m+1}-X_{N,j}^m}{\Delta t}=V_{N,j}^{m+b}\,,\\ 
& \frac{V_{N,j}^{m+1}-V_{N,j}^m}{\Delta t}= - E_h^{m+c}(X_{N,j}^{m+a})-\frac{V_{N,j}^++V_{N,j}^-}{2} \times B_h^{m}( X_{N,j}^{m})\,,\\
&V_{N,j}^{+}=V_{N,j}^{m+1}-\demi \Delta t E^{m}_h\left(X_{N,j}^{m}\right)\,,\\ 
&V_{N,j}^{-}=V_{N,j}^{m}+\demi \Delta t E^{m}_h\left(X_{N,j}^{m}\right)\,,
\end{align}
with $(a,b,c) \in \{0,1\}^3$. 

Numerical simulations show that the electric field and the velocity must be made implicit ($b=1$, $c=1$) to overcome the stability condition (\ref{cst-plasma}). They also show that, whatever the choice for $(a,b,c)$, the scheme remains subject to the Courant condition (\ref{cst-particle-velocity}). An explicit discretization of the position ($a=0$) is preferable to an implicit discretization ($a=1$), since it yields a scheme easier to solve and more accurate. In particular, in the case of a constant electric field and a zero magnetic field, the choice $a=0$ yields a symplectic scheme, unlike the choice $a=1$. Symplectic time-integration schemes ensure excellent conservation properties and an accurate behavior in long-time simulations \cite{HLW02}.

\section{KMC waves}\label{sec:KMC2}

The existence of KMC waves can be derived from the quasi-neutral equations (with motionless ions) presented in Section \ref{sec:QN:regime}. Neglecting the inertia term $\partial_t J$ and the pressure term $\nabla \cdot S$ in the generalized Ohm law \eqref{momentJ} and rewriting it with dimensional variables, we obtain the relation
\begin{equation*}
  n E - \frac{1}{e} J \times B = 0
\end{equation*}
Then, using the Maxwell-Amp\`ere equation and Maxwell-Faraday equations (still with dimensional variables), we deduce that the magnetic field evolution is governed only by the Hall term:
\begin{equation*}
  \partial_t B = - \nabla \times \big( ({\nabla \times  B \times B)/ (\mu_0  e n})  \big) \,.
\end{equation*}
In the two-dimensional setting \eqref{2DVM1}-\eqref{2DVM2}, if the density does not vary along the $x$-axis and does vary along the $y$-axis, the above equation simplifies into a Burgers-like nonlinear hyperbolic equation:
\begin{equation*}
  \partial_t B_z =- \frac{1}{2 e \mu_0} \partial_y \left(\frac{1}{n}\right) \partial_x B_z^2 \,.
\end{equation*}
This equation admits shock wave solutions, with a speed
\begin{equation}\label{eq:def:KMC:Speed}
  V = \frac{B_z}{e \mu_0 } \partial_y \left(\frac{1}{n}\right)\,.
\end{equation}
They are the so-called KMC waves, named after Kingsep, Mokhov and Chukbar.

\bibliographystyle{abbrv}
\bibliography{plasma}

\begin{thebibliography}{10}

\bibitem{acheritogaray_kinetic_2011}
M.~Acheritogaray, P.~Degond, A.~Frouvelle, and J.-G. Liu.
\newblock Kinetic formulation and global existence for the
  hall-magneto-hydrodynamics system.
\newblock {\em Kinetic and Related Models}, 4(4):901--918, Nov. 2011.

\bibitem{AGNR80}
J.~C. Adam, A.~Gourdin~Serveniere, J.-C. N{\'e}d{\'e}lec, and P.-A. Raviart.
\newblock Study of an implicit scheme for integrating {M}axwell's equations.
\newblock {\em Comput. Methods Appl. Mech. Engrg.}, 22(3):327--346, 1980.

\bibitem{Bar}
R.~Barthelm\'e.
\newblock {\em Le probl\`eme de de conservation de la charge dans le couplage
  des \'equations de Vlasov et de Maxwell.}
\newblock PhD thesis, Universit\'e de Strasbourg, 2005.

\bibitem{BCS07}
R.~Barthelm{\'e}, P.~Ciarlet, Jr., and E.~Sonnendr{\"u}cker.
\newblock Generalized formulations of {M}axwell's equations for numerical
  {V}lasov-{M}axwell simulations.
\newblock {\em Math. Models Methods Appl. Sci.}, 17(5):657--680, 2007.

\bibitem{BCDS09}
R.~Belaouar, N.~Crouseilles, P.~Degond, and E.~Sonnendr{\"u}cker.
\newblock An asymptotically stable semi-{L}agrangian scheme in the
  quasi-neutral limit.
\newblock {\em J. Sci. Comput.}, 41(3):341--365, 2009.

\bibitem{BiLa04}
C.~Birdsall and A.~Langdon.
\newblock {\em {Plasma physics via computer simulation}}.
\newblock Inst of Physics Pub Inc, 2004.

\bibitem{biskamp_nonlinear_1997}
D.~Biskamp.
\newblock {\em Nonlinear {Magnetohydrodynamics}}.
\newblock Cambridge University Press, July 1997.

\bibitem{Bor72}
J.~Boris.
\newblock {Relativistic plasma simulation-optimization of a hybrid code}.
\newblock In {\em Proceedings}, page~3. Naval Research Laboratory; for sale by
  the Supt. of Docs., US Govt. Print Off., 1972.

\bibitem{Bow01}
K.~Bowers.
\newblock {Implicit methods of solving the Maxwell equations suitable for
  particle-in-cell simulation of low temperature plasmas}.
\newblock {\em preprint}, 2001.

\bibitem{BrFo82}
J.~U. Brackbill and D.~W. Forslund.
\newblock An implicit method for electromagnetic plasma simulation in two
  dimensions.
\newblock {\em J. Comput. Phys.}, 46(2):271--308, 1982.

\bibitem{CDL08}
A.~Catella, V.~Dolean, and S.~Lanteri.
\newblock An unconditionally stable discontinuous galerkin method for solving
  the 2-d time-domain maxwell equations on unstructured triangular meshes.
\newblock {\em Magnetics, IEEE Transactions on}, 44(6):1250--1253, 2008.

\bibitem{Chen}
F.~Chen.
\newblock {\em {Introduction to plasma physics and controlled fusion: Plasma
  physics}}.
\newblock Plenum Pub Corp, 1984.

\bibitem{chen_energy-_2011}
G.~Chen, L.~Chac{\'o}n, and D.~C. Barnes.
\newblock An energy- and charge-conserving, implicit, electrostatic
  particle-in-cell algorithm.
\newblock {\em Journal of Computational Physics}, 230(18):7018--7036, Aug.
  2011.

\bibitem{chen_efficient_2012}
G.~Chen, L.~Chac{\'o}n, and D.~C. Barnes.
\newblock An efficient mixed-precision, hybrid {CPU}–{GPU} implementation of
  a nonlinearly implicit one-dimensional particle-in-cell algorithm.
\newblock {\em Journal of Computational Physics}, 231(16):5374--5388, June
  2012.

\bibitem{chen_fluid_2014}
G.~Chen, L.~Chac{\'o}n, C.~A. Leibs, D.~A. Knoll, and W.~Taitano.
\newblock Fluid preconditioning for newton–krylov-based, fully implicit,
  electrostatic particle-in-cell simulations.
\newblock {\em Journal of Computational Physics}, 258:555--567, Feb. 2014.

\bibitem{CLHP89}
B.~I. Cohen, A.~Langdon, D.~W. Hewett, and R.~J. Procassini.
\newblock Performance and optimization of direct implicit particle simulation.
\newblock {\em Journal of Computational Physics}, 81(1):151 -- 168, 1989.

\bibitem{CLF82}
B.~L. Cohen, A.~B. Langdon, and A.~Friedman.
\newblock Implicit time integration for plasma simulation.
\newblock {\em J. Comput. Phys.}, 46(1):15--38, 1982.

\bibitem{CDV07}
P.~Crispel, P.~Degond, and M.-H. Vignal.
\newblock An asymptotic preserving scheme for the two-fluid {E}uler-{P}oisson
  model in the quasineutral limit.
\newblock {\em J. Comput. Phys.}, 223(1):208--234, 2007.

\bibitem{CDV}
N.~Crouseilles, G.~Dimarco, and M.-H. Vignal.
\newblock {Multiscale schemes for the BGK-Vlasov-Poisson system in the
  quasi-neutral and fluid limits. Stability analysis and first order schemes,}.
\newblock submitted.

\bibitem{CRS09}
N.~Crouseilles, T.~Respaud, and E.~Sonnendrücker.
\newblock A forward semi-lagrangian method for the numerical solution of the
  vlasov equation.
\newblock {\em Computer Physics Communications}, 180(10):1730 -- 1745, 2009.

\bibitem{degond_asymptotic-preserving_2013}
P.~Degond.
\newblock Asymptotic-preserving schemes for fluid models of plasmas.
\newblock {\em {arXiv}:1104.1869}, 39-40:1--90, 2013.

\bibitem{DDNSV10}
P.~Degond, F.~Deluzet, L.~Navoret, A.-B. Sun, and M.-H. Vignal.
\newblock Asymptotic-preserving particle-in-cell method for the
  {V}lasov-{P}oisson system near quasineutrality.
\newblock {\em J. Comput. Phys.}, 229(16):5630--5652, 2010.

\bibitem{DDSprep}
P.~Degond, F.~Deluzet, and D.~Savelief.
\newblock {Numerical approximation of the Euler-Maxwell model in the
  quasineutral limit}.
\newblock {\em Journal of Computational Physics}, 231(4):1917 -- 1946, 2012.

\bibitem{Del03}
F.~Deluzet.
\newblock {Mathematical modeling of plasma opening switches.}
\newblock {\em Comput. Phys. Commun.}, 152(1):34--54, 2003.

\bibitem{DGAH10}
M.~Drouin, L.~Gremillet, J.-C. Adam, and A.~Héron.
\newblock Particle-in-cell modeling of relativistic laser–plasma interaction
  with the adjustable-damping, direct implicit method.
\newblock {\em Journal of Computational Physics}, 229(12):4781 -- 4812, 2010.

\bibitem{fernsler_quasineutral_2005}
R.~Fernsler, S.~Slinker, and G.~Joyce.
\newblock Quasineutral plasma models.
\newblock {\em Physical Review E}, 71(2), Feb. 2005.

\bibitem{fruchtman_sheath_1999}
A.~Fruchtman, J.~Grossmann, S.~Swanekamp, and P.~Ottinger.
\newblock Sheath propagation along the cathode of a plasma opening switch.
\newblock {\em {IEEE} Transactions on Plasma Science}, 27(5):1464--1468, Oct.
  1999.

\bibitem{GMAH08}
T.~Grismayer, P.~Mora, J.~Adam, and A.~H{\'e}ron.
\newblock Electron kinetic effects in plasma expansion and ion acceleration.
\newblock {\em Physical Review E}, 77(6):066407, 2008.

\bibitem{HLW02}
E.~Hairer, C.~Lubich, and G.~Wanner.
\newblock {\em Geometric numerical integration}, volume~31 of {\em Springer
  Series in Computational Mathematics}.
\newblock Springer-Verlag, Berlin, 2002.
\newblock Structure-preserving algorithms for ordinary differential equations.

\bibitem{HK11}
D.~Han~Kwan.
\newblock {\em Contribution \`a l'\'etude math\'ematique des plasmas fortement
  magn\'etis\'es}.
\newblock PhD thesis, Universit\'e Paris VI, 2011.

\bibitem{hewett_low-frequency_1994}
D.~W. Hewett.
\newblock Low-frequency electromagnetic ({D}arwin) applications in plasma
  simulation.
\newblock {\em Computer Physics Communications}, 84(1–3):243--277, Nov. 1994.

\bibitem{HeLa87}
D.~W. Hewett and A.~B. Langdon.
\newblock Electromagnetic direct implicit plasma simulation.
\newblock {\em Journal of Computational Physics}, 72(1):121 -- 155, 1987.

\bibitem{HoEa88}
R.~Hockney and J.~Eastwood.
\newblock {\em {Computer simulation using particles}}.
\newblock Taylor \& Francis, 1988.

\bibitem{Jin99}
S.~Jin.
\newblock {Efficient asymptotic-preserving (AP) schemes for some multiscale
  kinetic equations.}
\newblock {\em SIAM J. Sci. Comput.}, 21(2):441--454, 1999.

\bibitem{Joyce1997}
G.~Joyce, M.~Lampe, S.~P. Slinker, and W.~M. Manheimer.
\newblock Electrostatic particle-in-cell simulation technique for quasineutral
  plasma.
\newblock {\em Journal of Computational Physics}, 138(2):540--562, 1997.

\bibitem{Lan92}
A.~B. Langdon.
\newblock On enforcing gauss' law in electromagnetic particle-in-cell codes.
\newblock {\em Computer Physics Communications}, 70(3):447 -- 450, 1992.

\bibitem{LCF83}
A.~B. Langdon, B.~I. Cohen, and A.~Friedman.
\newblock Direct implicit large time-step particle simulation of plasmas.
\newblock {\em J. Comput. Phys.}, 51(1):107--138, 1983.

\bibitem{langmuir_interaction_1929}
I.~Langmuir.
\newblock The interaction of electron and positive ion space charges in cathode
  sheaths.
\newblock {\em Physical Review}, 33(6):954--989, 1929.

\bibitem{lapenta_particle_2012}
G.~Lapenta.
\newblock Particle simulations of space weather.
\newblock {\em Journal of Computational Physics}, 231(3):795--821, Feb. 2012.

\bibitem{lapenta_kinetic_2006}
G.~Lapenta, J.~U. Brackbill, and P.~Ricci.
\newblock Kinetic approach to microscopic-macroscopic coupling in space and
  laboratory plasmasa).
\newblock {\em Physics of Plasmas (1994-present)}, 13(5):055904, May 2006.

\bibitem{manfredi_vlasov_2011}
G.~Manfredi, S.~Hirstoaga, and S.~Devaux.
\newblock Vlasov modelling of parallel transport in a tokamak scrape-off layer.
\newblock {\em Plasma Physics and Controlled Fusion}, 53(1):015012, Jan. 2011.

\bibitem{markidis_energy_2011}
S.~Markidis and G.~Lapenta.
\newblock The energy conserving particle-in-cell method.
\newblock {\em Journal of Computational Physics}, 230(18):7037--7052, Aug.
  2011.

\bibitem{markidis_multi-scale_2010}
S.~Markidis, G.~Lapenta, and {Rizwan-uddin}.
\newblock Multi-scale simulations of plasma with {iPIC}3d.
\newblock {\em Mathematics and Computers in Simulation}, 80(7):1509--1519, Mar.
  2010.

\bibitem{Mas81}
R.~J. Mason.
\newblock Implicit moment particle simulation of plasmas.
\newblock {\em J. Comput. Phys.}, 41(2):233--244, 1981.

\bibitem{Mas87}
R.~J. Mason.
\newblock An electromagnetic field algorithm for {$2$}{D} implicit plasma
  simulation.
\newblock {\em J. Comput. Phys.}, 71(2):429--473, 1987.

\bibitem{RLB02}
P.~Ricci, G.~Lapenta, and J.~U. Brackbill.
\newblock A simplified implicit {M}axwell solver.
\newblock {\em J. Comput. Phys.}, 183(1):117--141, 2002.

\bibitem{richardson_particle--cell_2012}
A.~Richardson, D.~Hinshelwood, P.~Ottinger, J.~Schumer, S.~Swanekamp, and
  T.~Mehlhorn.
\newblock Particle-in-cell simulations of species separation in the plasma
  opening switch.
\newblock In {\em 2012 Abstracts {IEEE} International Conference on Plasma
  Science ({ICOPS})}, pages 1P--99--1P--99, July 2012.

\bibitem{schumer_mhd--pic_2001}
J.~Schumer, S.~Swanekamp, P.~Ottinger, R.~Commisso, B.~Weber, D.~N. Smithe, and
  L.~Ludeking.
\newblock {MHD}-to-{PIC} transition for modeling of conduction and opening in a
  plasma opening switch.
\newblock {\em {IEEE} Transactions on Plasma Science}, 29(3):479--493, June
  2001.

\bibitem{swanekamp_particlecell_1996}
S.~B. Swanekamp, J.~M. Grossmann, A.~Fruchtman, B.~V. Oliver, and P.~F.
  Ottinger.
\newblock Particle‐in‐cell simulations of fast magnetic field penetration
  into plasmas due to the hall electric field.
\newblock {\em Physics of Plasmas (1994-present)}, 3(10):3556--3563, Oct. 1996.

\bibitem{Taf}
A.~Taflove, S.~Hagness, et~al.
\newblock {\em {Computational electrodynamics: The finite-difference
  time-domain method}}.
\newblock Artech House Norwood, MA, 2000.

\bibitem{taitano_development_2013}
W.~Taitano, D.~Knoll, L.~Chac{\'o}n, and G.~Chen.
\newblock Development of a consistent and stable fully implicit moment method
  for vlasov--ampère particle in cell ({PIC}) system.
\newblock {\em {SIAM} Journal on Scientific Computing}, 35(5):S126--S149, Jan.
  2013.

\bibitem{VuBr92}
H.~Vu and J.~Brackbill.
\newblock Celest1d: an implicit, fully kinetic model for low-frequency,
  electromagnetic plasma simulation.
\newblock {\em Computer physics communications}, 69(2-3):253--276, 1992.

\bibitem{WBF86}
J.~M. Wallace, J.~U. Brackbill, and D.~W. Forslund.
\newblock An implicit moment electromagnetic plasma simulation in cylindrical
  coordinates.
\newblock {\em Journal of Computational Physics}, 63(2):434 -- 457, 1986.

\bibitem{yulan_particle--cell_2003}
C.~Yulan, Z.~Zhengzhong, and S.~Fengju.
\newblock Particle-in-cell simulations of plasma opening switch with external
  magnetic field.
\newblock 2003.

\bibitem{zeng_particle--cell_2001}
Z.~Zeng, Y.~Qiu, and E.~Kuffel.
\newblock A particle-in-cell simulation of plasma opening switch.
\newblock In L.~G. Christophorou and J.~K. Olthoff, editors, {\em Gaseous
  Dielectrics {IX}}, pages 155--160. Springer {US}, Jan. 2001.

\end{thebibliography}


\end{document}